\documentclass[a4paper,12pt]{article}
\usepackage{authblk}
\usepackage[utf8]{inputenc} 
\usepackage[english]{babel} 
\usepackage{lipsum} 
\usepackage{url} 
\usepackage{xcolor}
\usepackage{geometry} 
\usepackage{booktabs} 
\usepackage{fancyhdr} 
\usepackage{hyperref} 
\usepackage{amsmath} 
\usepackage[toc,page]{appendix}
\usepackage{bbold}
\usepackage{caption}
\usepackage{subcaption}
\usepackage{amsthm} 
\usepackage{mathtools} 
\usepackage{amsthm} 
\usepackage{bm} 
\usepackage{braket} 
\usepackage{tikz} 
\usepackage[compat=1.1.0]{tikz-feynman} 
\usepackage{pgfplots} 
\usepackage{graphicx}
\usepackage{xcolor}
\usepackage{float}
\newtheorem{theorem}{Theorem}

\usepackage{systeme}

\usepackage{amssymb} 
\usepackage{amsfonts} 
\newcommand{\R}{\mathbb{R}}

\newcommand{\de}{\partial}
\newcommand{\diff}{\mathrm{d}}
\newcommand{\dx}{\diff x}
\newcommand{\dt}{\diff t}

\newcommand{\dz}{\diff z}
\newcommand{\ds}{\diff s}

\newcommand{\Mean}[1]{\langle #1 \rangle}

\numberwithin{equation}{section}
\title{\textbf{Geometric Flow Equations for Schwarzschild-AdS Space-time and Hawking-Page Phase Transition}}
\author[a]{Davide De Biasio}
\author[a,b]{Dieter L\"ust}
\affil[a]{Arnold Sommerfeld Center for Theoretical Physics, \newline
Ludwig Maximilians Universit\"at M\"unchen, \newline Theresienstrasse 37, 80333 M\"unchen, Germany}
\affil[b]{Max--Planck--Institut f\"ur Physik, Werner--Heisenberg--Institut, \newline
F\"ohringer Ring 6, 80805 M\"unchen, Germany}

\begin{document}

\fancypagestyle{plain}{%
	\fancyhead[R]{LMU-ASC 19/20 \\
MPP-2020-80}
	\renewcommand{\headrulewidth}{0pt}
}
\maketitle

\abstract{Following
the recent observation that the Ricci flow and the infinite  distance swampland conjecture are closely related to each other, 
we will investigate in this paper geometric flow equations for AdS space-time geometries. First, 
we consider the so called Yamabe and  Ricci-Bourguignon flows and we show that these two flows - in contrast to the Ricci flow - 
lead to infinite distance fixed points for product spaces like $AdS_d\times S^p$, where $AdS_d$ denotes d-dimensional AdS space and $S^p$ corresponds to a p-dimensional sphere.
Second, we consider black hole geometries in AdS space time geometries and their behaviour under the Yamabe and  Ricci-Bourguignon flows.
Specifically we will examine if and how the AdS black holes will undergo a Hawking-Page phase transition under the Ricci flow, the
Yamabe flow and under the general Ricci-Bourguignon flow.}

\newpage
\tableofcontents
\newpage

\section{Introduction}
Geometric flow equations are of great interest in both mathematics as well in physics.
Mathematical flow equations in general relativity are certain differential equations, where one follows the flow of a family of metrics with respect to a certain path in field space. 
The most famous example is the Ricci flow, introduced by Hamilton  \cite{hamilton1982},
and subsequently investigated, among others, by Perelman \cite{Perelman:2006un}.
In  cosmology the Ricci-flow also plays an important role,
and furthermore Ricci flow has a close connection to renormalization group  flow in 2-dimensional string $\sigma$ model.\footnote{See e.g. \cite{Bakas:2005kv,Bakas:2007qm}  for a more detailed discussion about the 
relation between the geometric Ricci flow and the 2d RG flow.}

Recently a close connection between the 
swampland distance conjecture \cite{Ooguri:2006in} and geometric flow equations was observed \cite{Kehagias:2019akr}. (For a comprehensive review of the Swampland program, see \cite{Palti:2019pca}).
The  swampland  distance conjecture   states that at larges distances $\Delta$ in the field space of quantum gravity theories there must be an infinite tower of states with mass scale $m$ such that 
\begin{equation}
m = M_p e^{-\Delta } \,.
\label{dsc}
\end{equation}
($M_p=1/L_p$ is the Planck mass, and $L_p$ denotes the Planck length.)
The appearance of the massless tower of states in general implies that the effective field theory description breaks down above the mass scale $m$.
It was then argued \cite{Kehagias:2019akr} that  following the Ricci flow towards a fixed point, which is at infinite distance in the background space,
is accompanied by an infinite tower of states in quantum gravity.
Of course in concrete set-ups like in string theory, one has to ask  what is the origin of the massless tower states in terms of the microscopic theory.

 A second important point of the Ricci flow swampland conjecture 
 is that the entropy functionals of the gradient flow equations provide a sensible definition for the distance 
and hence for corresponding masses of the tower
of states along the flow.
For the case of the Ricci flow, the relevant distance can be defined in terms of the  scalar curvature of the background metric: $\Delta_R\simeq\log R$.

The Ricci flow conjecture was tested by several other examples. In particular
AdS space, where the AdS radius corresponds to the flow parameter of the Ricci flow equations, flows to flat space, which is at infinite distance.
In fact, for AdS spaces the Ricci flow swampland conjecture is equivalent to the anti-de Sitter distance conjecture (ADC)  \cite{Lust:2019zwm}, which  states 
  that the limit of a small AdS cosmological constant, $\Lambda\rightarrow0$,  is at infinite distance in the space of AdS metrics, and that it is related to infinite tower of states
with typical masses that scales as
\begin{equation}
{\rm ADC}:\quad m_{AdS}\sim\Lambda^{\alpha}\, ,\label{adstower}
\end{equation}
with $\alpha={\cal O}(1)$.  The strong version of the ADC proposes that for supersymmetric backgrounds  $\alpha=1/2$.
The corresponding distance is given in terms of the logarithm of $\Lambda$:
\begin{equation}
\Delta_{AdS}=-\alpha\log \Lambda\, .
\end{equation}   
The tower of states that emerge in the context of the ADC are very often given in terms Kaluza-Klein modes, which originate from an extra space time factor in addition to the AdS manifold.
The ADC then also implies that in the limit of small
AdS cosmological constant the AdS space cannot exist on its own in quantum gravity, but must be part of a higher-dimensional background manifold.
   
\vskip0.5cm
In this paper we like to discuss two kinds of generalizations and extensions of the Ricci flow conjecture.
First, in typical string theory constructions or in M-theory examples the total backgrounds are of the form
\begin{eqnarray}
M_d\times K^p\, ,
\end{eqnarray}
where $M_d$ often is given in terms of d-dimensional AdS space $AdS_d$ and $K^p$ corresponds to a p-dimensional sphere $S^p$.
As we will discuss, for these product manifolds the Ricci flow cannot provide a meaningful result, since there is no fixed point of the flow towards flat space.
Therefore one needs a refined version of the Ricci flow. As it turns out and as we will discuss, the socalled 
Yamabe flow (YF) is well suited for product manifolds of the above kind.
In contrast to the Ricci flow, which is driven by the Ricci tensor, the Yamabe flow is determined by the scalar curvature $R$ on the right-hand-side of the differential equation of the space-time metric.
Hence also for product spaces, there can be a non-trivial flow, which, as we will discuss,  extends to flat space at infinite distance.
Therefore one expects that also the endpoint of the Yamabe flow at infinite distance is accompanied by a massless tower of states.
This light tower of states, which typically correspond to the KK modes of $K^p$, arise in the flat limit of the total product space. Therefore  they should not be viewed 
as states that open up a new dimension of space-time, but as states that reconstruct flat (d+p)-dimensional space-time.

Another variant of a geometric flow is the Ricci-Bourguignon flow (RBF), which basically interpolated between the Ricci flow and the Yamabe flow.
We will see that like the Yamabe flow also the Ricci-Bourguignon flow is very useful to analyze product manifolds and will lead to non-trivial results.

\vskip0.5cm
The second topic of this paper is to consider geometric flow equations for black hole space-times. 
Similarly to the ADC and the Ricci conjecture for AdS geometries, it was discussed in \cite{Bonnefoy:2019nzv}
that the limit of large black entropies is at infinite distance in the space of black hole metrics and that this limit is accompanied by 
a tower of massless modes. However the standard Schwarzschild metric in asymptotically flat Minkowski space behaves trivially under
the Ricci flow, since the metric is Ricci flat. Therefore we will analyze the geometric flow equations for black holes in asymptotic AdS space-time geometries.\footnote{Ricci flow for black holes 
geometries was also discussed in \cite{Headrick:2006ti}.}
In fact, black holes in AdS space-time exhibit the 
interesting behaviour of Hawking-Page phase transition \cite{Hawking:1982dh}, which is correlated to the values of the black hole horizon and the AdS radius. 
We will discuss if and how the AdS black holes will undergo a Hawking-Page phase transition under the various geometric flows, namely under Ricci flow, the
Yamabe flow and under the general Ricci-Bourguignon flow. 
We will observe that the Haw\-king-Page phase transition always happens at a finite distance in the flow parameter, implying that there  should not be an infinite tower of particles to get massless  at the transition point, and the effective field theory description is  still valid at the Haw\-king-Page phase transition.

\section{Geometric Flows}
 \subsection{Ricci Flow (RF)}
 Let $\mathcal{M}$ be a smooth differentiable $d$-dimensional manifold, provided with a pseudo-Riemannian metric $g_{\mu\nu}$. Following \cite{hamilton1982}, we introduce \textit{Ricci Flow} equations
\begin{equation}
    \frac{\de g_{\mu\nu}}{\de\lambda}=-2B_{\mu\nu} \, ,
\end{equation}
where $\lambda$ is a \textit{flow parameter} and:
\begin{equation}
    B_{\mu\nu}\equiv R_{\mu\nu} \, .
\end{equation}
For a detailed discussion of Ricci flow, see \cite{Topping2006LecturesOT} and \cite{chow2004ricci}.
\subsubsection{AdS-Sphere product manifold}
A natural way to construct a (pseudo-)Riemannian metric structure on product manifolds is discussed in \ref{producto}.
\subsubsection{$AdS_{d}(\alpha)$ Ricci Flow}
Consider the  smooth manifold $\mathcal{M}\equiv AdS_{d}(\alpha)$, where $AdS_{d}(\alpha)$ is a $d$-dimensional Anti de Sitter spacetime with radius:
\begin{equation}
\alpha\equiv\sqrt{\frac{(d-2)(1-d)}{2\Lambda}} \, .
\end{equation}
Defining \textit{Poincarè coordinates} on half of $AdS_{d}$, the metric on $\mathcal{M}$ takes the form:
\begin{equation}
\ds^{2}=\frac{\alpha^{2}}{z^{2}}\biggl[\dz^{2}-\dt^{2}+\sum_{i=1}^{d-2}\dx_{i}^{2}\biggr] \, .
\end{equation}
It can be easily checked that the following identity is satisfied:
\begin{equation}
R_{\mu\nu}=-\frac{d-1}{\alpha^{2}}g_{\mu\nu}=\frac{2\Lambda}{d-2}g_{\mu\nu} \, .
\end{equation}
Taking $k$ to be the flow parameter, let's promote $\alpha$ to a $k$-dependent parameter $\alpha(k)$. Therefore, Ricci flow equations take the form
\begin{equation}
\frac{\de g_{\mu\nu}(k)}{\de k}=-2R_{\mu\nu}(k)=2\frac{d-1}{\alpha^2(k)}g_{\mu\nu}(k) \, ,
\end{equation}
which can be rephrased, more explicitly, as:
\begin{equation}
\frac{\de}{\de k}\biggl(\frac{\alpha^{2}(k)}{z^{2}}\biggr)=2\frac{d-1}{\alpha^2(k)}\frac{\alpha^{2}(k)}{z^{2}}\ \Longrightarrow\  \boxed{\alpha^{2}(k)=\alpha^{2}(k_{0})+2\bigl(d-1\bigr)\bigl(k-k_{0}\bigr)} \, .
\end{equation}
We now want to perform the same computation in \textit{global coordinates}, where the metric takes the form:
\begin{equation}
\ds^{2}=-\frac{\alpha^{2}+r^{2}}{\alpha^{2}}\dt^{2}+\frac{\alpha^{2}}{\alpha^{2}+r^{2}}\diff r^{2}+r^{2}\biggl[\diff\theta_{1}^{2}+\sum_{a=2}^{d-2}\diff\theta_{a}^{2}\prod_{b=1}^{a-1}\sin^{2}(\theta_{b})\biggr] \, .
\end{equation}
That said, we perform a further coordinate transformation and rephrase the metric as
\begin{equation}
d s^2=\alpha^2(-\cosh^2 \rho \, d \tau^2 + \, d \rho^2) + \alpha^{2}\sinh^2{\rho}\biggl[\diff\theta_{1}^{2}+\sum_{a=2}^{d-2}\diff\theta_{a}^{2}\prod_{b=1}^{a-1}\sin^{2}(\theta_{b})\biggr] \, ,
\end{equation}
where $r\equiv\alpha\sinh{\rho}$ and $t\equiv\alpha\tau$. Now, Ricci flow equations take the form
\begin{equation}
\frac{\de\alpha^{2}(k)}{\de k}=2\frac{d-1}{\alpha^2(k)}\alpha^{2}(k)
\end{equation}
for any metric component, since the only $k$-dependent parameter (i.e. $\alpha$) is \textit{factored out}, while in the previous expression it was non-trivially embedded in different ways into the various metric components. Hence, by taking care of this aspect, we derive the \textit{same}, exact result we obtained in Poincarè coordinates, namely:
\begin{equation}
\boxed{\alpha^{2}(k)=\alpha^{2}(k_{0})+2\bigl(d-1\bigr)\bigl(k-k_{0}\bigr)} \, .
\end{equation}
Therefore, Ricci flow equations can be consistenly derived either in global or Poincarè coordinates.

\subsubsection{$S^{p}(\rho)$ Ricci Flow}
Consider the smooth manifold $\mathcal{M}\equiv S^{p}(\rho)$, where $S^{p}(\rho)$ is a $p$-dimensional sphere with radius $\rho$.
Parametrizing $S^{p}$ with $p$ angles, the metric on $\mathcal{M}$ takes the form:
\begin{equation}
\ds^{2}=\rho^{2}\biggl[\diff\theta_{1}^{2}+\sum_{a=2}^{p}\diff\theta_{a}^{2}\prod_{b=1}^{a-1}\sin^{2}(\theta_{b})\biggr] \, .
\end{equation}
It can be easily checked that the following identity is satisfied:
\begin{equation}
R_{ab}=\frac{p-1}{\rho^2}\tilde{g}_{ab} \, .
\end{equation}
Let's now promote the radius $\rho$ to a $k$-dependent parameter $\rho(k)$. Therefore, Ricci flow equations take the form
\begin{equation}
\frac{\de\tilde{g}_{ab}(k)}{\de k}=-2\tilde{R}_{ab}(k)=-2\frac{p-1}{\rho^2(k)}\tilde{g}_{ab}(k) \, ,
\end{equation}
which can be rephrased, more explicitly, as:
\begin{equation}
\frac{\de\rho^{2}(k)}{\de k}=-2\frac{p-1}{\rho^2(k)}\rho^{2}(k) \ \Longrightarrow\  \boxed{\rho^{2}(k)=\rho^{2}(k_{0})-2\bigl(p-1\bigr)\bigl(k-k_{0}\bigr)} \, .
\end{equation}
\subsubsection{$AdS_{d}(\alpha)\times S^{p}(\rho)$ Ricci Flow}
Given the above considerations on Ricci flow for product manifolds, the flow of $AdS_{d}(\alpha)\times S^{p}(\rho)$ is simply described by:
\begin{equation}
\systeme*{\rho^{2}(k)=\rho^{2}(k_{0})-2\bigl(p-1\bigr)\bigl(k-k_{0}\bigr),\alpha^{2}(k)=\alpha^{2}(k_{0})+2\bigl(d-1\bigr)\bigl(k-k_{0}\bigr)} \, .
\end{equation}
Whether we flow towards increasing or decreasing values of $k$, we can't avoid encountering a singularity at a finite distance in flow parameter. Namely, we have $k_{1}<k<k_{2}$ such that:
\begin{equation}
\systeme*{\rho^{2}(k_{2})=0\longrightarrow k_{2}=k_{0}+\rho^{2}(k_{0})/\bigl(2p-2\bigr),\alpha^{2}(k_{1})=0\longrightarrow k_{1}=k_{0}-\alpha^{2}(k_{0})/\bigl(2d-2\bigr)} \, .
\end{equation}
Therefore, the parameters $\alpha$ and $\rho$ are forced to take values in finite ranges given by:
\begin{equation}
0<\alpha^{2}<\alpha^{2}(k_{2})=\alpha^{2}(k_{0})+\rho^{2}(k_{0})\frac{d-1}{p-1} \, ,
\end{equation}
\begin{equation}
0<\rho^{2}<\rho^{2}(k_{1})=\rho^{2}(k_{0})+\alpha^{2}(k_{0})\frac{p-1}{d-1} \, .
\end{equation}
\newpage
\newpage
\subsection{Yamabe Flow (YF)}\label{yamabe}
Let $\mathcal{M}$ be a smooth differentiable $d$-dimensional manifold, provided with a pseudo-Riemannian metric $g_{\mu\nu}$. We introduce \textit{Yamabe Flow} equations
\begin{equation}
    \frac{\de g_{\mu\nu}}{\de\lambda}=-2B_{\mu\nu} \, ,
\end{equation}
where $\lambda$ is a \textit{flow parameter} and:
\begin{equation}
    B_{\mu\nu}\equiv \frac{R}{d}g_{\mu\nu} \, .
\end{equation}
Information on the convergence properties of Yamabe flow can be found in \cite{2007InMat.170..541B}, \cite{2017arXiv170903192C} and \cite{Brendle2005ConvergenceOT}.

\subsubsection{AdS-Sphere product manifold}
Let's consider\footnote{We took the \textit{universal cover} of $AdS_{d}$ in order to get rid of closed timelike curves.} $AdS_{d}(\alpha)\times S^{p}(\rho)$ spacetime, with metric
\begin{equation}\label{metric}
\begin{split}
d s^2=&\alpha^2\left(-\cosh^2 \varrho \, d \tau^2 + \, d \varrho^2\right) + \alpha^{2}\sinh^2{\varrho}\left[\diff\theta_{1}^{2}+\sum_{a=2}^{d-2}\diff\theta_{a}^{2}\prod_{b=1}^{a-1}\sin^{2}(\theta_{b})\right]+\\
&+\rho^{2}\left[\diff\phi_{1}^{2}+\sum_{c=2}^{p}\diff\phi_{c}^{2}\prod_{d=1}^{c-1}\sin^{2}(\phi_{d})\right] \, ,
\end{split}
\end{equation}
where $\phi_{1},\dots,\phi_{p-1},\theta_{1},\dots,\theta_{d-3}\in[0,\pi]$, $\phi_{p},\theta_{d-2}\in[0,2\pi)$, $\tau\in\R$ and $\varrho\in[0,\infty)$.
\subsubsection{$AdS_{d}(\alpha)$ Yamabe Flow}
For $AdS_{d}(\alpha)$ spacetime, we have that $B_{\mu\nu}=R_{\mu\nu}$, therefore Yamabe Flow coincides with Ricci flow.
\subsubsection{$S^{p}(\rho)$ Yamabe Flow}
For $S^{p}(\rho)$ space, we have that $B_{\mu\nu}=R_{\mu\nu}$, therefore Yamabe Flow coincides with Ricci flow.
\subsubsection{$AdS_{d}(\alpha)\times S^{p}(\rho)$ Yamabe Flow}
If we consider $AdS_{d}(\alpha)\times S^{p}(\rho)$ spacetime, we have
\begin{equation}
    B_{\mu\nu}=\left[\frac{p(p-1)}{(d+p)\rho^{2}}-\frac{d(d-1)}{(d+p)\alpha^{2}}\right]g_{\mu\nu}
\end{equation}
and, therefore, Yamabe Flow equations take the form
\begin{equation}
    \frac{\de g_{\mu\nu}}{\de\lambda}=-2\left[\frac{p(p-1)}{(d+p)\rho^{2}}-\frac{d(d-1)}{(d+p)\alpha^{2}}\right]g_{\mu\nu} \, ,
\end{equation}
where $g_{\mu\nu}$, $\alpha$ and $\rho$ acquire a $\lambda$-dependence. Plugging in the explicit form of the metric, we have:
\begin{equation}
    \begin{split}
      \frac{\de \alpha^{2}}{\de\lambda}=&+2\frac{d(d-1)}{(d+p)} -2\frac{p(p-1)}{(d+p)}\frac{\alpha^{2}}{\rho^{2}}\equiv A(d,p)+B(d,p)\frac{\alpha^{2}}{\rho^{2}} \, , \\
      \frac{\de \rho^{2}}{\de\lambda}=&-2\frac{p(p-1)}{(d+p)}+2\frac{d(d-1)}{(d+p)}\frac{\rho^{2}}{\alpha^{2}}\equiv B(d,p)+A(d,p)\frac{\rho^{2}}{\alpha^{2}} \, .
    \end{split}
\end{equation}
First of all, we observe that:
\begin{equation}
    \frac{\de}{\de\lambda}\frac{\alpha^{2}}{\rho^{2}}=\frac{1}{\rho^{4}}\left[\rho^{2}\frac{\de \alpha^{2}}{\de\lambda}-\alpha^{2}\frac{\de \rho^{2}}{\de\lambda}\right]=0 \, .
\end{equation}
Therefore, we have
\begin{equation}
    \alpha^{2}(\lambda)=C_{0}\rho^{2}(\lambda) \, ,
\end{equation}
where $C_{0}$ is a real constant. Imposing this identity to the above equations, we can \textit{solve} them explicitly as:
\begin{equation}
    \begin{split}
        \alpha^{2}(\lambda)=&\alpha^{2}(\lambda_{0})+\left[A(d,p)+B(d,p)\frac{\alpha^{2}(\lambda_{0})}{\rho^{2}(\lambda_{0})}\right]\left(\lambda-\lambda_{0}\right) \, , \\
        \rho^{2}(\lambda)=&\rho^{2}(\lambda_{0})+\left[B(d,p)+A(d,p)\frac{\rho^{2}(\lambda_{0})}{\alpha^{2}(\lambda_{0})}\right]\left(\lambda-\lambda_{0}\right) \, .
    \end{split}
\end{equation}
With a lighter notation, imposing $\lambda_{0}=0$, we obtain:
\begin{equation}
        \begin{split}
        \alpha^{2}(\lambda)=&C_{0}\rho^{2}_{0}+\left[A(d,p)+B(d,p)C_{0}\right]\lambda \, , \\
        \rho^{2}(\lambda)=&\rho^{2}_{0}+\left[B(d,p)+A(d,p)/C_{0}\right]\lambda \, .
    \end{split}
\end{equation}
Given the values of $d$ and $p$, we have that:
\begin{itemize}
    \item $\alpha^{2}$ and $\rho^{2}$ \textit{grow} along the flow for $C_{0}<K\equiv\frac{d(d-1)}{p(p-1)}$;
    \item $\alpha^{2}$ and $\rho^{2}$ \textit{stay fixed} along the flow for $C_{0}=K$;
    \item $\alpha^{2}$ and $\rho^{2}$ \textit{decrease} along the flow for $C_{0}>K$.
\end{itemize}
In this product manifold example, the two radii {always} have the same flow behaviour.
Then we observe that for $C_{0}<K$\footnote{Namely, for large enough sphere radius.}, the flow has an {infinite distance fixed point}, where the $AdS$ term tends to Minkowski spacetime and the \textit{internal} sphere radius tends to infinity. Therefore, potential \textit{Kaluza-Klein} modes get {light}, as we expected from Swampland conjectures. 
This light tower of states corresponds to the KK modes of the sphere and  arises in the flat limit of the total product space. Therefore  they should not be viewed 
as states that open up a new dimension of space-time, but as states that reconstruct flat (d+p)-dimensional space-time.

\subsubsection{General flow behaviour}
In order to investigate general properties of the flow, we observe that
\begin{equation}
    \frac{\de g^{\mu\nu}}{\de\lambda}=2\frac{R}{d}g^{\mu\nu} 
\end{equation}
and deduce the following equation:
\begin{equation}
    \frac{\de R}{\de\lambda}=R_{\mu\nu}\frac{\de g^{\mu\nu}}{\de\lambda}+\nabla_{\rho}\left(g^{\sigma\nu}\frac{\de\Gamma^{\rho}_{\nu\sigma}}{\de\lambda}-g^{\sigma\rho}\frac{\de\Gamma^{\mu}_{\mu\sigma}}{\de\lambda}\right) \, .
\end{equation}
We separately compute the $\lambda$-derivative of Christoffel's symbols as:
\begin{equation}
    \begin{split}
        \frac{\de\Gamma^{\rho}_{\nu\sigma}}{\de\lambda}=&\frac{1}{2}\frac{\de}{\de\lambda}\biggl[g^{\rho\alpha}\left(\de_{\nu}g_{\alpha\sigma}+\de_{\sigma}g_{\alpha\nu}-\de_{\alpha}g_{\nu\sigma}\right)\biggr]=\\
        =&\frac{1}{2}\biggl[2\frac{R}{d}g^{\rho\alpha}\left(\de_{\nu}g_{\alpha\sigma}+\de_{\sigma}g_{\alpha\nu}-\de_{\alpha}g_{\nu\sigma}\right)\biggr]+\\
        &+\frac{1}{2}g^{\rho\alpha}\biggl(\de_{\nu}\frac{\de g_{\alpha\sigma}}{\de\lambda}+\de_{\sigma}\frac{\de g_{\alpha\nu}}{\de\lambda}-\de_{\alpha}\frac{\de g_{\nu\sigma}}{\de\lambda}\biggr)=\\
        =&-\frac{1}{d}\biggl(\delta^{\rho}_{\ \sigma}\de_{\nu}R+ \delta^{\rho}_{\ \nu}\de_{\sigma}R-g_{\nu\sigma}g^{\rho\alpha}\de_{\alpha}R\biggr) \, .
    \end{split}
\end{equation}
Therefore, we have:
\begin{equation}
    \begin{split}
        g^{\sigma\nu}\frac{\de\Gamma^{\rho}_{\nu\sigma}}{\de\lambda}-g^{\sigma\rho}\frac{\de\Gamma^{\mu}_{\mu\sigma}}{\de\lambda}=&
        g^{\sigma\rho}\frac{1}{d}\biggl(\delta^{\mu}_{\ \sigma}\de_{\mu}R+ \delta^{\mu}_{\ \mu}\de_{\sigma}R-g_{\mu\sigma}g^{\mu\alpha}\de_{\alpha}R\biggr)+\\
        &-g^{\sigma\nu}\frac{1}{d}\biggl(\delta^{\rho}_{\ \sigma}\de_{\nu}R+ \delta^{\rho}_{\ \nu}\de_{\sigma}R-g_{\nu\sigma}g^{\rho\alpha}\de_{\alpha}R\biggr)=\\
        =&\frac{1}{d}\biggl(g^{\mu\rho}\de_{\mu}R+d g^{\sigma\rho}\de_{\sigma}R-g^{\alpha\rho}\de_{\alpha}R\biggr)+\\
        &-\frac{1}{d}\biggl(g^{\rho\nu}\de_{\nu}R+g^{\sigma\rho} \de_{\sigma}R-dg^{\rho\alpha}\de_{\alpha}R\biggr)=\\
        =&2\frac{d-1}{d}g^{\rho\mu}\de_{\mu}R \, .
    \end{split}
\end{equation}
Being $R$ a scalar, we have: 
\begin{equation}
    \de_{\mu}R=\nabla_{\mu}R \, .
\end{equation}
Hence, by defining $\Box\equiv\nabla^{\mu}\nabla_{\mu}$, our flow equation for $R$ becomes:
\begin{equation}
\frac{\de R}{\de\lambda}=\frac{2}{d}R^{2}+2\frac{d-1}{d}\Box R \, .
\end{equation}
By defining
\begin{equation}
    K\equiv\frac{R}{d-1}
\end{equation}
and
\begin{equation}
    \xi\equiv\frac{2(d-1)^{2}}{d}\lambda \, ,
\end{equation}
we obtain:
\begin{equation}\label{scalarflow}
\boxed{\frac{\de K}{\de\xi}=K^{2}+\Box K} \, .
\end{equation}
\subsubsection{Constant Ricci scalar example}
Considering the case of a manifold with constant $R$, like $AdS_{d}(\alpha)\times S^{p}(\rho)$, our flow equation \ref{scalarflow} reduces to
\begin{equation}
    \frac{\de K}{\de\xi}=K^{2}
\end{equation}
and can be solved by
\begin{equation}
    K(\xi)=\frac{K_{0}}{1-K_{0}\xi} \, ,
\end{equation}
where we have chosen, without loss of generality, $\xi_{0}=0$ and defined $K(\xi_{0})\equiv K_{0}$. Introducing $R$ and $\lambda$ again, we have:
\begin{equation}
    R(\lambda)=\frac{R_{0}}{1-2R_{0}\frac{d-1}{d}\lambda} \, .
\end{equation}
By redefining the flow parameter as $\mu\equiv 2(d-1)\lambda/d$, we have:
\begin{equation}
        R(\mu)=\frac{R_{0}}{1-R_{0}\mu} \, .
\end{equation}
Using the \textit{scalar curvature} definition of the distance along the flow, as introduced in \cite{Kehagias:2019akr}, we have:
\begin{equation}
    \Delta(\mu)\propto\log\left(\frac{R_{0}}{R(\mu)}\right)=\log\left(1-R_{0}\mu\right) \, .
\end{equation}
Using, on the other hand, the standard definition of a path length on the metrics manifold, extensively discussed in \cite{1992math......1259G}, we obtain the same exact result
\begin{equation}
\begin{split}
    \tilde{\Delta}(\mu)\propto&\int_{0}^{\mu}\diff\bar{\mu}\left(\frac{1}{V_{\mathcal{M}}}\int_{\mathcal{M}}\sqrt{g}g^{MN}g^{OP}\frac{\de g_{MO}}{\de\bar{\mu}}\frac{\de g_{NP}}{\de\bar{\mu}}\right)^{\frac{1}{2}}=\\
    =&\int_{0}^{\mu}\diff\bar{\mu}\left(\frac{1}{V_{\mathcal{M}}}\int_{\mathcal{M}}\sqrt{g}g^{MN}g^{OP}\frac{\de g_{MO}}{\de\bar{\lambda}}\frac{\de g_{NP}}{\de\bar{\lambda}}\right)^{\frac{1}{2}}\frac{\de\bar{\lambda}}{\de\bar{\mu}}=\\
    \propto&\int_{0}^{\mu}\diff\bar{\mu}|R(\bar{\mu})|=\pm\log\left(1-R_{0}\mu\right) \, ,
\end{split}
\end{equation}
where "+" corresponds to $R_{0}<0$ and "-" corresponds to $R_{0}>0$. 
Therefore, depending on the choice of $R_{0}$, we have two possible non-trivial flow behaviours:
\begin{itemize}
    \item $R_{0}>0$: The flow encounters a \textbf{singularity}, since $R\rightarrow\infty$ as $\mu\rightarrow 1/R_{0}$, with $\Delta(\mu=1/R_{0}|R_{0}>0)=\infty$.
    \item $R_{0}<0$: The flow tends to an \textbf{infinite distance fixed point}, since $R\rightarrow 0$ as $\mu\rightarrow\infty$, with $\Delta(\mu=\infty|R_{0}<0)=\infty$.
\end{itemize}
The discussion is trivially reversed for $\mu$ decreasing towards negative values. Therefore, both the case in which $R_{0}>0$ and the one in which $R_{0}<0$ have an \textit{infinite distance fixed point} and a \textit{singularity}, which is at an infinite distance on the manifold of all metrics despite being at a finite distance in the flow parameter. As discussed in \cite{Kehagias:2019akr}
the presence of these infinite distance flow fixed-points is of relevance in probing the swampland conjectures.

\newpage
\subsection{Ricci-Bourguignon Flow (RBF)}
Let $\mathcal{M}$ be a smooth differentiable $d$-dimensional manifold, provided with a pseudo-Riemannian metric $g_{\mu\nu}$. We introduce \textit{Ricci-Bourguignon flow} equations
\begin{equation}\label{ricicib}
    \frac{\de g_{\mu\nu}}{\de\lambda}=-2C_{\mu\nu}(\gamma) \, ,
\end{equation}
as defined in \cite{bourguignon1981ricci}, where $\lambda$ is a \textit{flow parameter}, $\gamma\in\R$ and:
\begin{equation}
    C_{\mu\nu}(\gamma)\equiv R_{\mu\nu}-\gamma Rg_{\mu\nu} \, .
\end{equation}
For a detailed review of the subject, see \cite{article}.\\
First of all, we see that we are dealing with a family of geometric flows, labelled by $\gamma$, interpolating between \textit{Ricci flow} and \textit{Yamabe flow}. In particular, it should be observed that there are some \textit{special values} of $\rho$ for which $C_{\mu\nu}$ takes interesting forms:
\begin{itemize}
    \item $\gamma=0$: $C_{\mu\nu}$ is the \textit{Ricci tensor} $R_{\mu\nu}$;
    \item $\gamma=1/2(d-1)$: $C_{\mu\nu}$ is the \textit{Schouten tensor} $R_{\mu\nu}-Rg_{\mu\nu}/2(d-1)$;
    \item $\gamma=1/d$: $C_{\mu\nu}$ is the \textit{traceless Ricci tensor} $R_{\mu\nu}-Rg_{\mu\nu}/d$;
    \item $\gamma=1/2$: $C_{\mu\nu}$ is the \textit{Einstein tensor} $R_{\mu\nu}-Rg_{\mu\nu}/2$;
    \item $\gamma\rightarrow-\infty$: $C_{\mu\nu}$ is, after a suitable rescaling, the \textit{Yamabe tensor} $-Rg_{\mu\nu}/d$.
\end{itemize}
The \textit{scalar curvature} flow behaviour can be derived as:
\begin{equation}\label{scalar}
    \boxed{\frac{\de R}{\de\lambda}=-2\gamma R^{2}+2R_{\mu\nu}R^{\mu\nu}+\left(1-2(d-1)\gamma\right)\Box R} \, .
\end{equation}
\subsubsection{Einstein manifold Ricci-Bourguignon Flow}\label{einstein}
Consider a d-dimensional manifold for which
\begin{equation}
    R_{\mu\nu}=\Omega g_{\mu\nu}
\end{equation}
and, subsequently:
\begin{equation}
    R=\Omega d \, .
\end{equation}
In such a framework, we have
\begin{equation}
    C_{\mu\nu}(\gamma)= \left(1-\gamma d\right)\Omega g_{\mu\nu}
\end{equation}
and the flow equation takes the form
\begin{equation}
    \frac{\de g_{\mu\nu}}{\de\lambda}=-2\left(1-\gamma d\right)\Omega g_{\mu\nu} \, ,
\end{equation}
which can be solved explicitly as:
\begin{equation}
    g_{\mu\nu}(\lambda)=g_{\mu\nu}(\lambda_{0})\exp{\left(-2\left(1-\gamma d\right)\int_{\lambda_{0}}^{\lambda}\diff\xi\Omega(\xi)\right)} \, .
\end{equation}
The equation \ref{scalar} for the flow behaviour of Ricci scalar takes the form:
\begin{equation}
    \frac{\de\Omega}{\de\lambda}=2\left(1-\gamma d\right)\Omega^{2} 
\end{equation}
and can be solved as:
\begin{equation}
    \boxed{\Omega(\lambda)=\frac{\Omega_{0}}{1+2\left(\gamma d-1\right)\Omega_{0}(\lambda-\lambda_{0})}} \, .
\end{equation}
Therefore, the metric behaviour for \textit{Einstein manifolds} under \textit{Ricci-Bourguignon} flow is
\begin{equation}
    \boxed{g_{\mu\nu}(\lambda)=g_{\mu\nu}(\lambda_{0})\left[1+2\left(\gamma d-1\right)\Omega_{0}(\lambda-\lambda_{0})\right]} \, ,
\end{equation}
which is nothing more than a \textit{linear Weyl rescaling}. From now on, without loss of generality, we'll assume $\lambda_{0}=0$.
Except for the trivial case in which $\gamma=1/d$, flow fixed points correspond to $\Omega=0$. Therefore, starting from a generic $\Omega_{0}$, we look for a flow parameter value $\bar{\lambda}$ such that $\Omega(\bar{\lambda})=0$. By defining $\varepsilon\equiv 2\left(\gamma d-1\right)\Omega_{0}$, we describe two distinct scenarios:
\begin{itemize}
    \item $\varepsilon>0$: the flow has a \textbf{fixed point} for $\lambda\rightarrow +\infty$ and a \textbf{singularity} for $\lambda\rightarrow -1/\varepsilon$;
    \item $\varepsilon<0$: the flow has a \textbf{fixed point} for $\lambda\rightarrow -\infty$ and a \textbf{singularity} for $\lambda\rightarrow -1/\varepsilon$.
\end{itemize}
As discussed in section \ref{yamabe}, distance along the flow can be computed as:
\begin{equation}
\begin{split}
    \Delta(\lambda)\propto&\log\left(\frac{R_{0}}{R(\lambda)}\right)=\log\left(\frac{\Omega_{0}}{\Omega(\lambda)}\right)=\\
    =&\log\left(1+2\left(\gamma d-1\right)\Omega_{0}\lambda\right)=\\
    =&\log\left(1+\varepsilon\lambda\right) \, .
    \end{split}
\end{equation}
Hence, the \textbf{fixed point} and the \textbf{singular point}, for both $\varepsilon>0$ and $\varepsilon<0$, are at an \textbf{infinite distance} on the manifold of all metrics, despite the latter being at a finite distance in the flow parameter.
\subsubsection{$AdS_{d}(\alpha)\times S^{p}(\rho)$ Ricci-Bourguignon Flow}
Let's consider, once more, the $AdS_{d}(\alpha)\times S^{p}(\rho)$ spacetime manifold, endowed with the natural product manifold metric discussed in \ref{producto}. Therefore, we have:
\begin{equation}
    C_{MN}(\gamma)=R_{MN}-\gamma\left[\frac{p(p-1)}{\rho^{2}}-\frac{d(d-1)}{\alpha^{2}}\right]g_{MN} \, .
\end{equation}
Considering the metric expression \ref{metric} and following the usual steps, our flow equations for $\alpha$ and $\rho$ are
\begin{equation}
\begin{split}
    \frac{\de\alpha^{2}}{\de\lambda}=&2(d-1)+2\gamma\left[p(p-1)\frac{\alpha^{2}}{\rho^{2}}-d(d-1)\right]=\\
    =&2(d-1)[1-\gamma d]+2\gamma p(p-1)\frac{\alpha^{2}}{\rho^{2}}
\end{split}
\end{equation}
and:
\begin{equation}
\begin{split}
    \frac{\de\rho^{2}}{\de\lambda}=&-2(p-1)+2\gamma\left[p(p-1)-d(d-1)\frac{\rho^{2}}{\alpha^{2}}\right]=\\
    =&-2(p-1)[1-\gamma p]-2\gamma d(d-1)\frac{\rho^{2}}{\alpha^{2}} \, .
    \end{split}
\end{equation}
We define $m\equiv 2(d-1)[1-\gamma d]$, $n\equiv 2\gamma p(p-1)$, $k\equiv -2(p-1)[1-\gamma p]$, $q\equiv -2\gamma d(d-1)$, $f(\lambda)\equiv\alpha^{2}(\lambda)/m$ and $g(\lambda)\equiv\rho^{2}(\lambda)/k$, obtaining:
\begin{equation}
\begin{split}
    f'(\lambda)=&1+\frac{n}{k}\frac{f(\lambda)}{g(\lambda)}\equiv 1+A\frac{f(\lambda)}{g(\lambda)} \, , \\
    g'(\lambda)=&1+\frac{q}{m}\frac{g(\lambda)}{f(\lambda)}\equiv 1+B\frac{g(\lambda)}{f(\lambda)} \, .
\end{split}
\end{equation}
In spite of their harmless appearance, the above equations are rather complicated to be solved exactly. Therefore, we provide some examples.\\
\subsubsection{Specific examples}
First, we impose to work in the typical string theory \textit{compactification} example with $d=4$ and $p=6$. Hence, we have:
\begin{equation}
    A=-\frac{6\gamma}{1-6\gamma};\qquad B=-\frac{4\gamma}{1-4\gamma} \, .
\end{equation}
Now, we impose $\lambda_{0}=0$, $\alpha^{2}(0)=1$, $\rho^{2}(0)=5$ initial conditions\footnote{This way, we are in the initial conditions regime for which Yamabe flow provides a growing behaviour for both AdS and the sphere.} and produce flow behaviour plots for different values of $\gamma$.
\begin{figure}[H] 
  \begin{subfigure}[b]{0.5\linewidth}
    \centering
        \includegraphics[width=\linewidth]{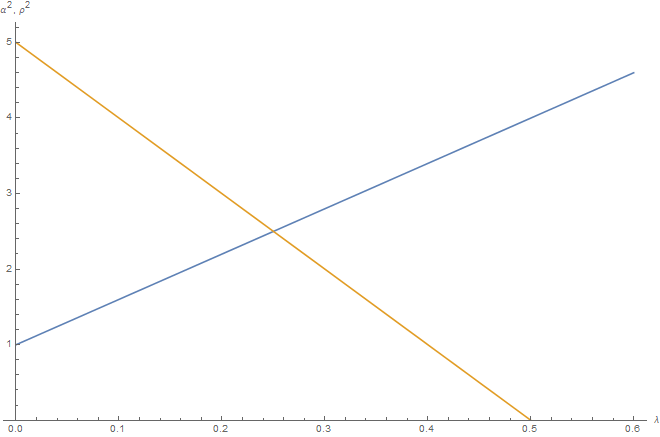}
    \caption{$\gamma=0$}
    \vspace{4ex}
  \end{subfigure}
  \begin{subfigure}[b]{0.5\linewidth}
    \centering
    \includegraphics[width=\linewidth]{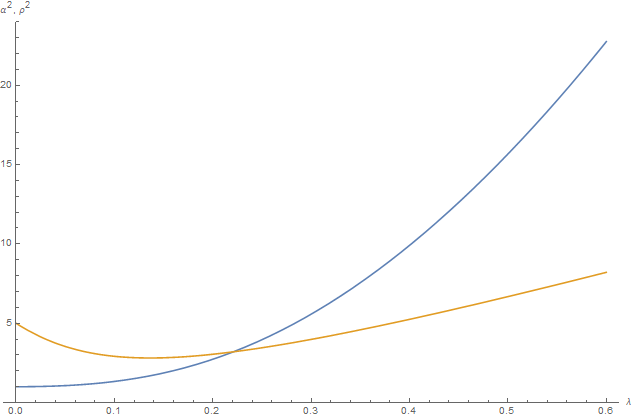}
    \caption{$\gamma=0.5$}
    \vspace{4ex}
  \end{subfigure}
      \begin{subfigure}[b]{0.5\linewidth}
    \centering
        \includegraphics[width=\linewidth]{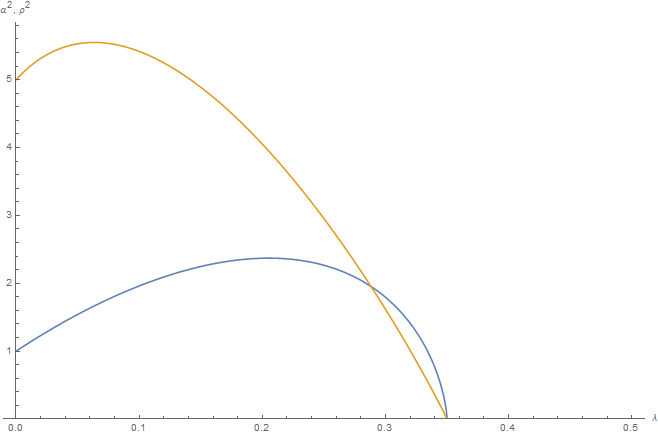}
    \caption{$\gamma=-0.5$}
    \vspace{4ex}
  \end{subfigure}
  \begin{subfigure}[b]{0.5\linewidth}
    \centering
    \includegraphics[width=\linewidth]{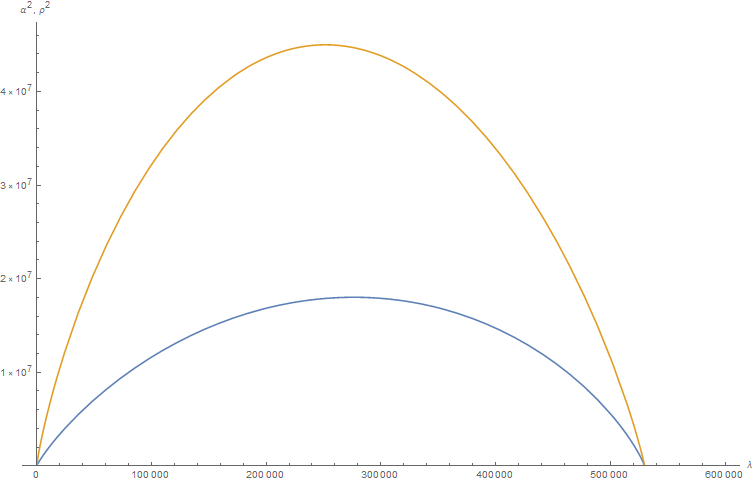}
    \caption{$\gamma=-30$}
    \vspace{4ex}
  \end{subfigure} 
\end{figure}
In the above graphs, the orange plot corresponds to $\rho^{2}$ and the blue one to $\alpha^2$.

\section{Schwarzschild-AdS black hole}\label{sadschap}
The following discussion serves as a fast review of some significant thermodynamical properties of Schwarzschild-AdS black hole spacetime. For a standard reference on thermodynamical properties of black holes, see \cite{Hawking:1976de}. For a specific discussion of black holes thermodynamics in AdS spacetime, see \cite{Hawking:1982dh} and \cite{Socolovsky_2018}. Information on the holographic interplay between Hawking-Page phase transition of Schwarzschild-AdS black holes and confinement in thermal gauge theories can be found in \cite{Witten:1998zw}.

\subsection{Schwarzschild-AdS thermodynamics}

Let's consider an asymptotically-$AdS_{4}(\alpha)$, non-rotating and uncharged \textit{black hole} solution to Einstein field equations
\begin{equation}
\ds^{2}=-\biggl(1 - \frac{2M}{r}+ \frac{r^2}{\alpha^{2}} \biggr)\dt^2 + \biggl( 1 - \frac{2M}{r}+ \frac{r^2}{\alpha^{2}} \biggr)^{-1}\diff r^2 + r^2 \diff\Omega^2_{2} \, ,
\end{equation}
where we chose to work in natural units and set $3\alpha^{2}\equiv-\Lambda$.
According to \textit{Generalized Birkhoff's theorem} for static, spherically symmetric vacuum solutions to Einstein field equations with negative cosmological constant, the above metric is unique up to diffeomorphisms. We define
\begin{equation}
\Delta(r;M,\alpha)\equiv 1 - \frac{2M}{r}+ \frac{r^2}{\alpha^{2}}
\end{equation}
and observe that, for $M>0$, such a black hole possesses an outer event horizon at $r_{+}$, defined as the largest root of $\Delta$, such that the metric has a coordinate singularity at $r=r_{+}$. The exact form for $r_{+}$ is given by:
\begin{equation}
r_{+}=\left(M\alpha^{2}\right)^{1/3}\left[\left(1+\sqrt{1+\frac{\alpha^{2}}{27M^{2}}}\right)^{1/3}+\left(1-\sqrt{1+\frac{\alpha^{2}}{27M^{2}}}\right)^{1/3}\right] \, .
\end{equation}
In order to compute \textit{Hawking temperature} for the outer horizon, we perform a Wick rotation $t\mapsto\tau\equiv it$ and move to Euclidean coordinates $(\tau,\rho,\theta,\phi)$, where $\theta$ and $\phi$ are the usual angular coordinates on $S^{2}$ and $\rho$ is defined by:
\begin{equation}
r\equiv r_{+}+\rho^{2} \, .
\end{equation}
Since we are interested in near-horizon behaviour, we can take $\rho^{2}\ll r_{+}$ and keep only leading contributions in $\rho^{2}$. By doing so, the metric takes the approximate form:
\begin{equation}
\ds^{2}\approx\frac{4\alpha^{2}r_{+}}{\alpha^{2}+3r_{+}^{2}}\biggl[\biggl(\frac{\alpha^{2}+3r_{+}^{2}}{2\alpha^{2}r_{+}}\biggr)^{2}\rho^{2}\diff\tau^{2}+\diff\rho^{2}\biggr]+r_{+}^{2}\diff\Omega_{2}^{2} \, .
\end{equation}
In order to avoid \textit{conical singularities} at $r=r_{+}$, we must impose
\begin{equation}
\bar{\tau}\equiv\biggl(\frac{\alpha^{2}+3r_{+}^{2}}{2\alpha^{2}r_{+}}\biggr)\tau
\end{equation}
to be $2\pi$-periodic. Therefore, $\tau$ has to be periodic with period $\beta$, where
\begin{equation}
\beta=\frac{4\pi\alpha^{2}r_{+}}{\alpha^{2}+3r_{+}^{2}} \, ,
\end{equation}
from which we can easily read off Hawking temperature as
\begin{equation}
T_{H}=\beta^{-1}=\frac{\alpha^{2}+3r_{+}^{2}}{4\pi\alpha^{2}r_{+}} \, ,
\end{equation}
where $r_{+}$, as always, has to be intended as a function of $M$ and $\alpha$.
\begin{center}
\begin{figure}[H]
\centering
  \includegraphics[width=0.6\linewidth]{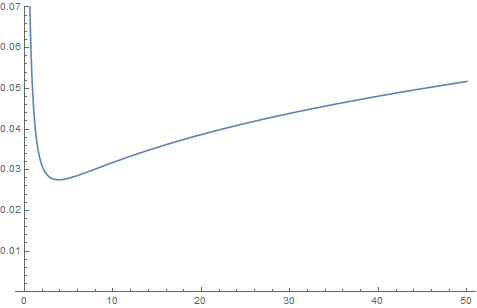}
\caption{Schwarzschild-AdS Hawking temperature as a function of $M$, in which $\alpha$ is chosen to have fixed value $10$.}
\label{fig:test}
\end{figure}
\end{center}
The main feature of {Schwarzschild-AdS black holes}, distinguishing them from asymptotically flat Schwarzschild solutions, is that they can only exist for $T$ bigger or equal to a limit temperature , that can be computed as:
\begin{equation}
    T_{0}=\frac{\sqrt{3}}{2}\frac{1}{\pi\alpha} \, .
\end{equation}
In particular, we have a single possible mass value for $T=T_{0}$ and two different branches for $T>T_{0}$, which we'll refer to as \textit{small} and \textit{large} black hole solutions, respectively. While \textit{small} SAdS black holes are thermodynamically unstable, \textit{large} SAdS black holes have positive heat capacity. Hence, they can be put in a stable thermal equilibrium with a background radiation-filled AdS reservoir.

\subsubsection{Hawking-Page phase transition}
In the following section our aim is to compare thermodynamical quantities characterizing pure radiation-filled AdS and Schwarzschild-AdS spacetime, in order to figure out which portions of the parameter space correspond to the former being energetically favoured and which ones, on the other hand, to the latter. In order to exclude the trivial case in which no black hole can exist, we assume to work with:
\begin{equation}
T>T_{0} \, .
\end{equation}
First of all, let's consider the path integral expression for the Euclidean quantum partition function of our gravity theory
\begin{equation}
Z=\int\mathcal{D}[g]e^{-I[g]} \, ,
\end{equation}
where:
\begin{equation}
I[g]=\frac{1}{16\pi}\int\diff^{4}x\sqrt{-g}\biggl(R[g]-2\Lambda\biggr) \, .
\end{equation}
The main contribution can be singled out in the \textit{saddle point} approximation
\begin{equation}
Z\approx e^{-I[g_{0}]} \, ,
\end{equation}
where $I[g_{0}]$ is nothing more than a global minimum of the action\footnote{Here, we do not consider the presence of additional saddle points introducing non-perturbative corrections.}. Helmoltz free energy, whose minimum corresponds to the energetically-favoured configuration of the system, can be approximated by:
\begin{equation}
F\equiv-T\log Z\approx TI[g_{0}] \, .
\end{equation}
Since we want to compare thermal AdS (tAdS) and Schwarzschild-AdS (SAdS) at the same temperature $T$, we only have to compute their respective actions $I$, in which a non-trivial temperature-dependent behaviour will be shown to appear. Hence, the one with the smallest value of $I(T)$ will be the energetically favoured one. For both AdS and Schwarzschild-AdS, we choose to introduce a radial cut-off in order to make the action integral finite. Therefore, we have
\begin{equation}
I_{1}=\frac{\Lambda}{8\pi}\int_{0}^{\tau_{1}}\dt\int_{0}^{K}r^{2}\diff r\int_{S^{2}}\diff\Omega_{2}^{2}=\frac{\Lambda}{6}\tau_{1}K^{3}
\end{equation}
for thermal AdS and 
\begin{equation}
I_{0}=\frac{\Lambda}{8\pi}\int_{0}^{\tau_{0}}\dt\int_{r_{+}}^{K}r^{2}\diff r\int_{S^{2}}\diff\Omega_{2}^{2}=\frac{\Lambda}{6}\tau_{0}\biggl(K^{3}-r_{+}^{3}\biggr)
\end{equation}
for Schwarzschild-AdS, where the radial integration starts from the outer horizon radius $r_{+}$. This is due to the fact that physics on the outside, which is the one we want to compare, can't be influenced by what happens at $r<r_{+}$. $\tau_{0}$ is fixed by $\alpha$ and $M$, while $\tau_{1}$ can be fixed by imposing time coordinates to have the same thermal periodicity. Thus, we have:
\begin{equation}
\tau_{1}=\tau_{0}\sqrt{\frac{\alpha^{2}K-2M\alpha^{2}+K^{3}}{\alpha^{2}K+K^{3}}} \, .
\end{equation}
For $r_{+}\ll K$, we obtain:
\begin{equation}
\Delta I\equiv I_{0}-I_{1}\approx\frac{\pi r_{+}^{2}\bigl(\alpha^{2}-r_{+}^{2}\bigr)}{\alpha^{2}+3r_{+}^{2}} \, .
\end{equation}
From a naive perspective, $T\Delta I$ is supposed to encode the free energy contribution due to the presence of a black hole in Schwarzschild-AdS. We can check our intuition by applying the usual formulas from statistical mechanics
\begin{equation}
\langle E \rangle=\frac{\de}{\de\beta_{0}}\Delta I=M,\qquad S=\beta_{0}\langle E \rangle-\Delta I=\frac{A_{BH}}{4}
\end{equation}
and deriving the 
Hawking-Page phase transition critical temperature
\begin{equation}
T_{C}=\frac{1}{\pi\alpha}\, .
\end{equation}
Specifically, we discover that:
\begin{enumerate}
\item $\Delta I<0\Leftrightarrow \alpha^{2}<r_{+}^{2}\Leftrightarrow T>T_{C}$ Schwarzschild-AdS is energetically favoured. Thermal AdS can exist, but it can reduce its free energy by tunnelling to Schwarzschild-AdS;
\item $\Delta I>0\Leftrightarrow \alpha^{2}>r_{+}^{2}\Leftrightarrow T<T_{C}$ thermal AdS is energetically favoured. Schwarzschild-AdS can exist, but it can reduce its free energy by tunnelling to thermal AdS;
\item $\Delta I=0\Leftrightarrow \alpha^{2}=r_{+}^{2}\Leftrightarrow T=T_{C}$ we are at a critical point, where Schwarzschild-AdS and thermal AdS are equally favoured.
\end{enumerate}
From the previous discussion, we can observe that Hawking-Page phase  temperature is consistent with the lower boundary on black hole temperature, namely
$
T_{C}>T_{0}
$.
In summary, starting from a temperature $T_{0}<T<T_{C}$ and increasing it, we encounter a \textit{first order phase transition} at $T=T_{C}$, named \textit{Hawking-Page} phase transition, where large Schwarzschild-AdS becomes preferred over thermal AdS.\\

 \newpage
 \newpage
\subsection{Schwarzschild-AdS Ricci-Bourguignon Flow}
In order to compute \textit{Ricci-Bourguignon flow} equations\footnote{In this section, we are going to discuss the general $\gamma$-behaviour of Schwarzschild-AdS metric under Ricci-Bourguignon flow equations. 
By imposing $\gamma=0$, one can get back to the standard Ricci flow equations.} for Schwarzschild-AdS spacetime, it is convenient to move to a coordinate system in which the unphysical divergence at the horizon, $r=r_{+}$, is removed. Therefore, following \cite{Socolovsky_2018}, we start by defining a \textit{tortoise radius}:
\begin{equation}
\begin{split}
\bar{r}\equiv&\int_{0}^{r}\dx\biggl(1 - \frac{2M}{x}+ \frac{x^2}{\alpha^{2}}  \biggr)^{-1}=\\
=&\frac{\alpha^{2}}{3r_{+}^{2}+\alpha^{2}}\biggl[r_{+}\log\biggl|1-\frac{r}{r_{+}}\biggr|-\frac{r_{+}}{2}\log\biggl(1+\frac{r(r+r_{+})}{r_{+}^{2}+\alpha^{2}}\biggr)+\\
&+\frac{3r_{+}^{2}+2\alpha^{2}}{\sqrt{3r_{+}^{2}+4\alpha^{2}}}\arctan\biggl(\frac{r\sqrt{3r_{+}^{2}+4\alpha^{2}}}{2(r_{+}^{2}+\alpha^{2})+rr_{+}}\biggr)\biggr] \, .
\end{split}
\end{equation} 
Furthermore, we introduce \textit{null coordinates}
\begin{equation}
u\equiv t-\bar{r},\qquad v\equiv t+\bar{r}
\end{equation}
and obtain:
\begin{equation}
\ds^{2}=-\biggl( 1 - \frac{2M}{r} + \frac{r^2}{\alpha^{2}} \biggr)\diff u\diff v + r^2 \diff\Omega^2_{2} \, .
\end{equation}
Now, we move to \textit{Kruskal-like} coordinates
\begin{equation}
\begin{split}
U&\equiv -\exp\bigl(-2\pi T_{H}u\bigr)=-\exp\biggl(-\frac{\alpha^{2}+3r_{+}^{2}}{2\alpha^{2}r_{+}}u\biggr)\, , \\
V&\equiv\exp\bigl(2\pi T_{H}v\bigr)=\exp\biggl(\frac{\alpha^{2}+3r_{+}^{2}}{2\alpha^{2}r_{+}}v\biggr)
\end{split}
\end{equation}
and get:
\begin{equation}
\begin{split}
\ds^{2}=&\biggl( 1 - \frac{2M}{r} + \frac{r^2}{\alpha^{2}} \biggr)\frac{1}{4\pi^{2}T_{H}^{2}}\frac{\diff U\diff V}{UV} + r^2 \diff\Omega^2_{2}=\\
=&-\biggl[ 1 - \frac{r_{+}}{r}\biggl(1+\frac{r_{+}^{2}}{\alpha^{2}}\biggr)+ \frac{r^2}{\alpha^{2}} \biggr]\frac{4\alpha^{4}r_{+}^{2}}{(\alpha^{2}+3r_{+}^{2})^{2}}\biggl|1-\frac{r}{r_{+}}\biggr|^{-1}\biggl(1+\frac{r(r+r_{+})}{r_{+}^{2}+\alpha^{2}}\biggr)^{1/2}\\&\exp\biggl[-\frac{3r_{+}^{2}+2\alpha^{2}}{r_{+}\sqrt{3r_{+}^{2}+4\alpha^{2}}}\arctan\biggl(\frac{r\sqrt{3r_{+}^{2}+4\alpha^{2}}}{2(r_{+}^{2}+\alpha^{2})+rr_{+}}\biggr)\biggr]\diff U\diff V + r^2 \diff\Omega^2_{2}=\\
=&-G(r;\alpha,r_{+})\diff U\diff V + r^2 \diff\Omega^2_{2} \, .
\end{split}
\end{equation}
Assuming to work with $r>r_{+}$, we have
\begin{equation}
\begin{split}
G(r;\alpha,r_{+})=&\frac{4\alpha^{2}r_{+}^{3}(\alpha^{2}+r^{2}+r_{+}^{2}+rr_{+})^{3/2}}{r(\alpha^{2}+r_{+}^{2})^{1/2}(\alpha^{2}+3r_{+}^{2})^{2}}\exp\biggl[-\frac{3r_{+}^{2}+2\alpha^{2}}{r_{+}\sqrt{3r_{+}^{2}+4\alpha^{2}}}\\&\arctan\biggl(\frac{r\sqrt{3r_{+}^{2}+4\alpha^{2}}}{2(r_{+}^{2}+\alpha^{2})+rr_{+}}\biggr)\biggr]=\\
\equiv&\rho(r;r_{+},\alpha)\alpha^{2}\exp\biggl[-\mu(r_{+};r_{+},\alpha)^{-1}\arctan\mu(r;r_{+},\alpha)\biggr] \, ,
\end{split}
\end{equation}
where:
\begin{equation}\label{rho}
\rho(r;r_{+},\alpha)\equiv\frac{4r_{+}^{3}(\alpha^{2}+r^{2}+r_{+}^{2}+rr_{+})^{3/2}}{r(\alpha^{2}+r_{+}^{2})^{1/2}(\alpha^{2}+3r_{+}^{2})^{2}} \, ,
\end{equation}
\begin{equation}\label{mu}
\mu(r;r_{+},\alpha)\equiv\frac{r\sqrt{3r_{+}^{2}+4\alpha^{2}}}{2(r_{+}^{2}+\alpha^{2})+rr_{+}} \, .
\end{equation}
From four-dimensional vacuum Einstein's equations with negative cosmological constant $\Lambda$, as defined above, we trivially have:
\begin{equation}
R_{\mu\nu}=-\frac{3}{\alpha^{2}}g_{\mu\nu} \, .
\end{equation}
Hence, we are dealing with an Einstein manifold with
\begin{equation}
    \Omega(\lambda)=\frac{-3}{\alpha^{2}(\lambda)}
\end{equation}
and \textit{Ricci-Bourguignon flow} equations, as discussed in \ref{einstein}, can be solved by
\begin{equation}
G(r;\lambda)=G(r;\lambda_{0})\exp{\left[(1-\gamma d)\int_{\lambda_{0}}^{\lambda}\frac{6}{\alpha^{2}(\lambda')}\diff\lambda'\right]} \, ,
\end{equation}
where:
\begin{equation}
    \alpha^{2}(\lambda)=\alpha^{2}(\lambda_{0})+6(1-\gamma d)(\lambda-\lambda_{0}) \, .
\end{equation}
It is important to stress the fact that, being the flow equation \ref{scalar} for the scalar curvature of Schwarzschild-AdS spacetime formally equivalent to the one for the scalar curvature of pure-AdS, the flow behaviour of $\alpha$ coincides with the one derived for pure-AdS.
Therefore, we have:
\begin{equation}
    \frac{G(r;\lambda)}{\alpha^{2}(\lambda)}=\frac{G(r;\lambda_{0})}{\alpha^{2}(\lambda_{0})} \, .
\end{equation}
By defining $r_{+}\equiv r_{+}(\lambda)$, $\alpha\equiv \alpha(\lambda)$, $\bar{r}_{+}\equiv r_{+}(\lambda_{0})$ and $\bar{\alpha}\equiv \alpha(\lambda_{0})$ and by plugging the expression for $G$ into Ricci flow equations, we obtain:
\begin{equation}\label{rbf}
\boxed{
\log\frac{\rho(r;r_{+},\alpha)}{\rho(r;\bar{r}_{+},\bar{\alpha})}+\frac{\arctan\mu(r;\bar{r}_{+},\bar{\alpha})}{\mu(\bar{r}_{+};\bar{r}_{+},\bar{\alpha})}-\frac{\arctan\mu(r;r_{+},\alpha)}{\mu(r_{+};r_{+},\alpha)}=0} \, .
\end{equation}
Before moving to the explicit discussion of the flow equations, it is crucial to define the \textit{radii ratio} parameter 
\begin{equation}
    k\equiv\frac{r_{+}}{\alpha} \, ,
\end{equation}
which will serve as an \textit{order parameter} for Hawking-Page phase transition:
\begin{itemize}
    \item $k>1$: $T>T_{C}$, the black hole phase is dominant;
    \item $k<1$: $T<T_{C}$, pure-AdS phase is dominant.
\end{itemize}
Therefore, we are interested in studying the flow behaviour of $k$.
\subsubsection{Solution behaviour}
In the following section, we are going to discuss the flow behaviour of $k$, directly descending from that of $r_{+}$, for different initial conditions $\bar{k}$ and different values of $\gamma$. In all the following graphs, the dashed lines correspond to the \textit{Hawking-Page} phase transition occurring at $k=1$. Furthermore, the flow behaviours are plotted at a fixed value of $r$: changing it doesn't seem to alter the qualitative properties of the flow in interesting ways\footnote{For example, the singularity of the $\gamma=0.5$-case doesn't depend on $r$!}, but it can create computational problems near $\lambda=0$.
\newpage
\begin{figure}[H]
\centering
\includegraphics[width=0.6\linewidth]{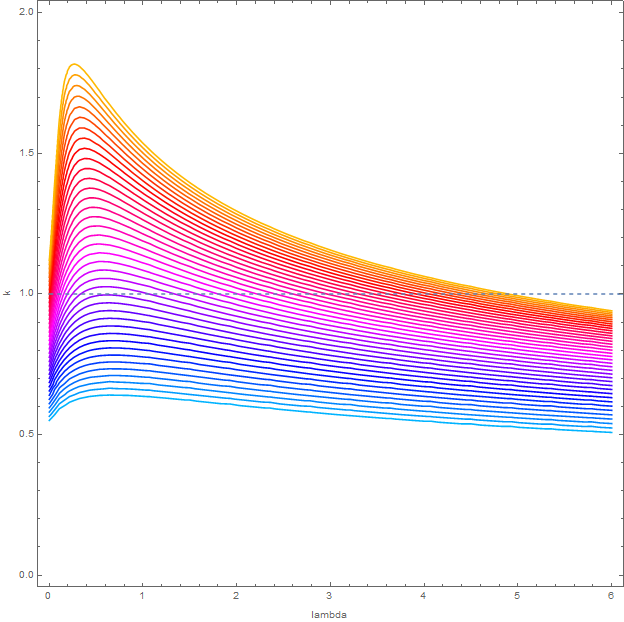}
\caption{$k(\lambda)$ flow behaviour at $r=3.5$, with $\gamma=0$, $\bar{\alpha}=1$ and $\bar{k}\in[0.55,1.12]$}
    \label{one}
    \vspace{4ex}
    \includegraphics[width=0.6\linewidth]{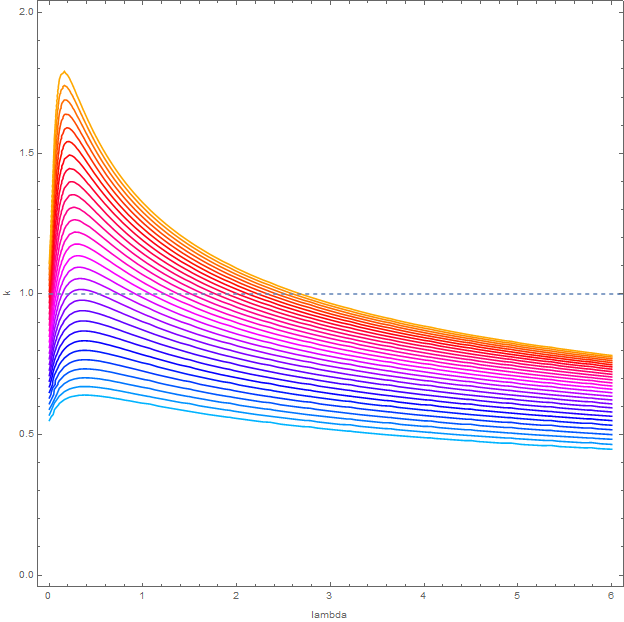}
    \caption{$k(\lambda)$ flow behaviour at $r=3.5$, with $\gamma=-0.2$, $\bar{\alpha}=1$ and $\bar{k}\in[0.55,1.12]$}
    \label{two}
\end{figure}
\begin{figure}[H]
\centering
\includegraphics[width=0.6\linewidth]{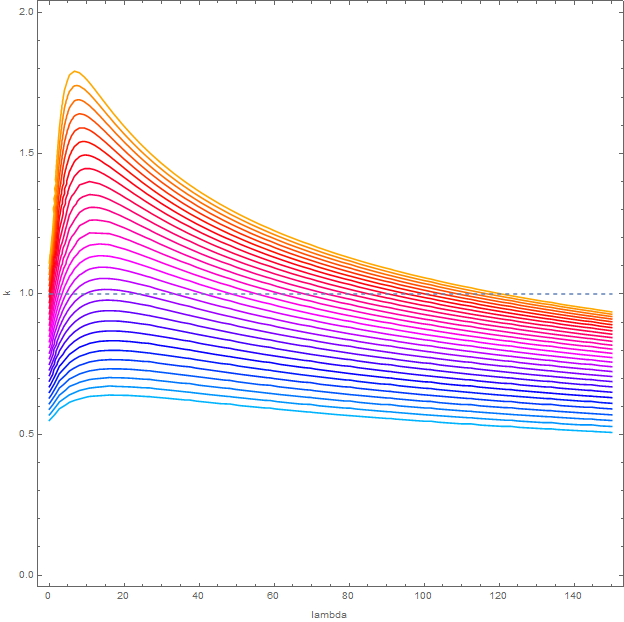}
    \caption{$k(\lambda)$ flow behaviour at $r=3.5$, with $\gamma=0.24$, $\bar{\alpha}=1$ and $\bar{k}\in[0.55,1.12]$}
    \label{three}
    \vspace{4ex}
\centering
\includegraphics[width=0.6\linewidth]{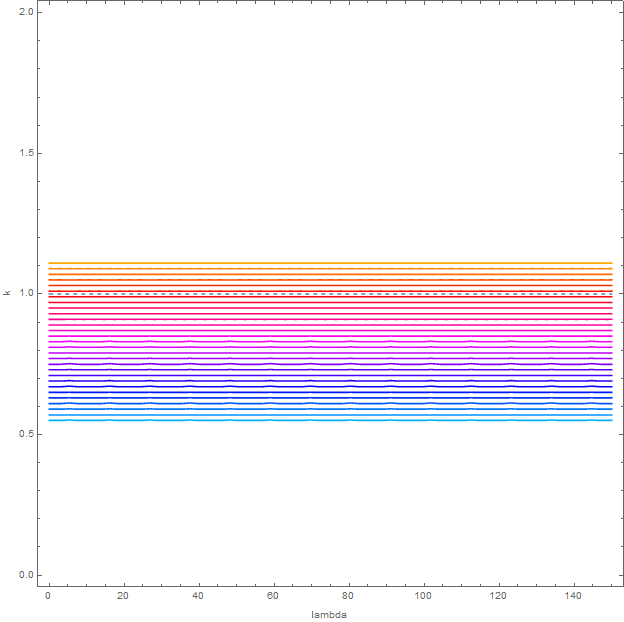}
    \caption{$k(\lambda)$ flow behaviour at $r=3.5$, with $\gamma=0.25$, $\bar{\alpha}=1$ and $\bar{k}\in[0.55,1.12]$}
    \label{four}
\end{figure}
\begin{figure}[H]
\centering
\includegraphics[width=0.6\linewidth]{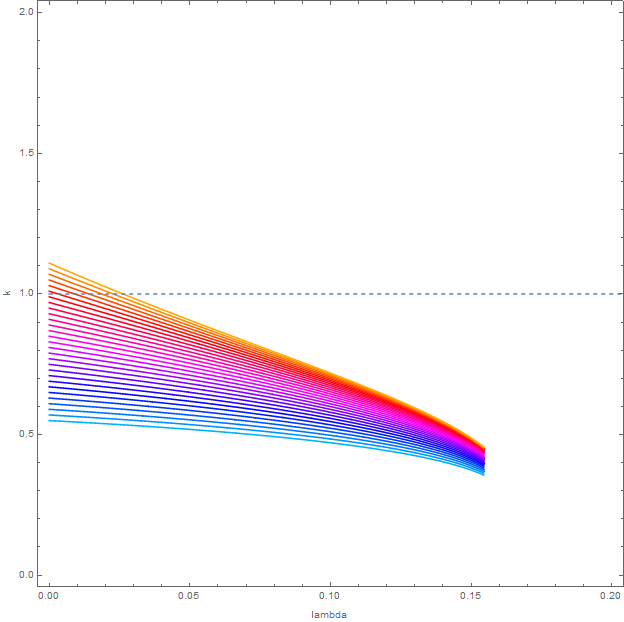}
    \caption{$k(\lambda)$ flow behaviour at $r=3.5$, with $\gamma=0.5$, $\bar{\alpha}=1$ and $\bar{k}\in[0.55,1.12]$}
    \label{five}
    \vspace{4ex}
\centering
\includegraphics[width=0.6\linewidth]{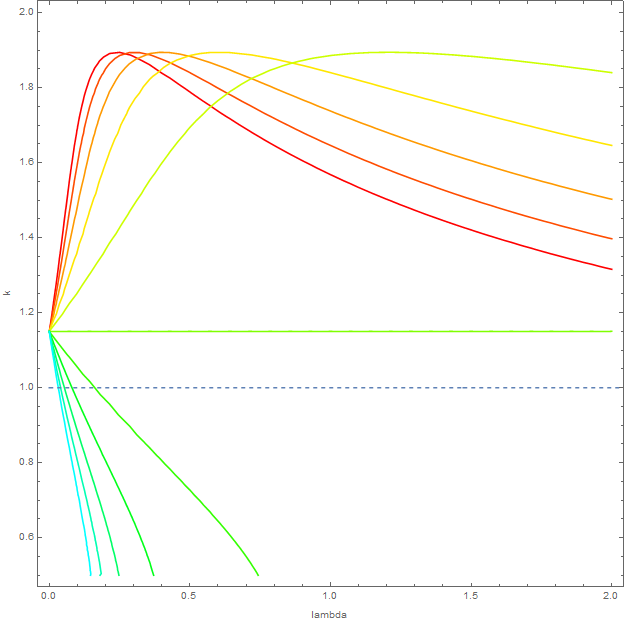}
    \caption{$k(\lambda)$ flow behaviour at $r=3.5$, with $\gamma\in[0,0.5]$, $\bar{\alpha}=1$ and $\bar{k}=1.15$}
    \label{six}
\end{figure}
Observing the above plots, we can discuss\footnote{The following statements were checked numerically for very large values of the flow para\-meter $\lambda$ and resulted to be really solid. An analytic approach might provide us with formal proofs.} some general properties of Ricci-Bourguignon flow solutions for Schwarzschild-AdS spacetime:
\begin{enumerate}
    \item Picture \ref{one}, which corresponds to the $\gamma=0$ case, shows us that \textit{Ricci flow} can induce Hawking-Page phase transition. In particular, there is a critical value $\varepsilon<0$ such that:
    \begin{itemize}
        \item For $\bar{k}<\varepsilon$, the system stays in the $k<1$ phase for the whole trajectory, never encountering the critical value;
        \item For $\varepsilon<\bar{k}<0$, the system undergoes the phase transition twice;
        \item For $\bar{k}>0$, the system undergoes the phase transition once.
    \end{itemize}
    We can observe that we \textit{always} end up with $k<1$. Namely, our spacetime flows towards the pure-AdS phase, staying in the black hole phase for \textit{at most} a finite interval in the flow parameter.
    \item Pictures \ref{two} and \ref{three} show us that small variations in $\gamma$, either towards negative or positive values, don't alter the qualitative behaviour of the flow significantly. The only effect this tuning of $\gamma$ has is shrinking or dilating the graph.
    \item Picture \ref{four} is particularly interesting, since it serves as an important \textit{consistency check} for our computations. Indeed, in $4$ dimensions, choosing $\gamma=0.25$ corresponds to removing the trace of the $C_{\mu\nu}$ tensor on the RHS of \ref{ricicib}: being our metric diagonal, Ricci-Bourguignon flow equations become trivial and can only be solved by fixed-point solutions. The fact that our numerical approximation program correctly solves this fully-understood case increases our confidence in its accuracy.
    \item For $\gamma>1/d$, $\alpha(\lambda)$ \textit{decreases} along the flow, reaching zero at a finite distance $\bar{\lambda}-\lambda_{0}$, with:
    \begin{equation}
        \alpha(\bar{\lambda})=0,\qquad\bar{\lambda}=\lambda_{0}+\frac{\alpha^{2}(\lambda_{0})}{6(\gamma d-1)} \, .
    \end{equation}
    In four dimensions, we have $1/d=0.25$. This behaviour can be clearly observed in \ref{five}, where our flow reaches a \textit{finite-distance singularity} where $\alpha^{2}=0$.
    If we, therefore, interested in testing \textit{Swampland conjectures}, it is important to take $\gamma>1/d$.
\end{enumerate}
From our numerically computed flow trajectories, we observe that the Haw\-king-Page phase transition, at least for the values of $\gamma$ taken into account\footnote{Which, indeed, cover the most significant and insightful cases.}, always happens at a finite distance in the flow parameter $\lambda$. This means that according to the swampland \textit{distance conjecture}, there  should not be an infinite tower of particles to get massless  at the transition point and the effective field theory description is still valid at the Haw\-king-Page phase transition.

\subsubsection{Yamabe flow}
Yamabe flow corresponds to a suitable rescaling of Ricci-Bourguignon flow, after having taken the limit $\gamma\rightarrow -\infty$. For Schwarzschild-AdS spacetime, we know that:
\begin{equation}
    R_{\mu\nu}=-\frac{3}{\alpha^{2}}g_{\mu\nu} \, .
\end{equation}
Hence, we are dealing with an Einstein manifold, for which:
\begin{equation}
    R_{\mu\nu}=\frac{R}{d}g_{\mu\nu} \, .
\end{equation}
Therefore, Ricci flow \textit{coincides} with Yamabe flow. Namely, the limit $\gamma\rightarrow-\infty$ produces a flow behaviour which is fully equivalent to the one discussed for the $\gamma=0$ case.

\section{Conclusions}
After having reviewed some properties of \textit{Ricci}, \textit{Yamabe} and \textit{Ricci-Bourguig\-non} geometric flow equations, a few concrete examples were discussed in detail. In particular, in the $AdS_{d}\times S^{p}$ space-time compactification case, it was shown that, in a reasonable initial parameters range, the  \textit{Yamabe} and the \textit{Ricci-Bourguignon}  flows end in a fixed point that it is at infinite distance along the 
trajectory of the flow parameters.
This behaviour, which differentiates \textit{Yamabe} flow solutions from \textit{Ricci} flow solutions, is consistent with what we would expect from the swampland \textit{distance conjecture}, but in the sense
that the corresponding tower  should not be viewed 
as states that open up a new dimension of space-time, but as states that reconstruct the flat (d+p)-dimensional space-time.


\vskip0.2cm

\noindent
In section \ref{sadschap}, some well-known thermodynamical properties of \textit{Schwarzschild-AdS} solution to Einstein field equations were revised. Subsequently, \textit{Ricci-Bourguig\-non} flow equations for Schwarzschild-AdS were written explicitly in a compact form and solved numerically for different values of $\gamma$. By doing so, we discovered that our geometric flow trajectories could induce \textit{Hawking-Page} phase transitions. On top of that, it was observed that, after a finite length in the flow parameter, the system tends to end up in the \textit{thermal AdS} phase.
In the appendices, we solved \textit{geodesic equations} for Schwarzschild-AdS metric explicitly.
We  also observed that the Haw\-king-Page phase transition always happens at a finite distance in the flow parameter, such that there  should not be an infinite tower of massless particles   at the transition point.

\newpage
\begin{appendices}
\section{Product manifolds}\label{producto}
\begin{theorem}
Let $\mathcal{M}$ and $\mathcal{N}$ be two differentiable smooth manifolds, respectively parametrized by local coordinates $\{x^{\mu}\}_{\mu=1}^{d}$ and $\{y^{a}\}_{a=1}^{p}$. Let $m_{\mu\nu}(x)$ be a metric on $\mathcal{M}$ and let $n_{ab}(y)$ be a metric on $\mathcal{N}$, with Ricci tensors $M_{\mu\nu}(x)$ and $N_{ab}(y)$.\footnote{We assume to work with metric-compatible Levi-Civita connections.} Therefore, the product manifold $\mathcal{R}\equiv\mathcal{M}\times\mathcal{N}$ equipped with the metric 
\begin{equation}
r_{MN}(x,y)\equiv\begin{bmatrix}
m_{\mu\nu}(x) & 0 \\
0 & n_{ab}(y)
\end{bmatrix}
\end{equation}
has Ricci tensor:
\begin{equation}
R_{MN}(x,y)=\begin{bmatrix}
M_{\mu\nu}(x) & 0 \\
0 & N_{ab}(y)
\end{bmatrix} \, .
\end{equation}
\end{theorem}
\begin{proof}
First, we compute the Christoffel's symbols $\Gamma^{M}_{NP}$ for the product manifold $\mathcal{R}$, where $\Omega^{\alpha}_{\mu\nu}$ and $\Theta^{a}_{bc}$ are, respectively, the Christoffel's symbols for $\mathcal{M}$ and $\mathcal{N}$. We list all possible cases:
\begin{itemize}
    \item $\Gamma^{\mu}_{\alpha\beta}=\frac{1}{2}r^{M\mu}\bigl(r_{M\beta,\alpha}+r_{M\alpha,\beta}-r_{\alpha\beta,M}\bigr)=\Omega^{\mu}_{\alpha\beta}$
    \item $\Gamma^{a}_{bc}=\frac{1}{2}r^{Ma}\bigl(r_{Mb,c}+r_{Mc,b}-r_{bc,M}\bigr)=\Theta^{a}_{bc}$
    \item $\Gamma^{a}_{b\mu}=\Gamma^{a}_{\nu\mu}=\Gamma^{\mu}_{b\nu}=\Gamma^{\mu}_{ab}=0$
\end{itemize}
Hence, we can trivially deduce that $R_{MN}$ satisfies the expression given above.
\end{proof}
The previous result dramatically simplifies the analysis of Ricci flow equations for a product manifold $\mathcal{R}$, since it allows us to treat the factors $\mathcal{M}$ and $\mathcal{N}$ separately.
Furthermore, we observe that $\mathcal{G}$ has \textit{Ricci scalar}
\begin{equation}
    R\equiv g^{MN}R_{MN}=M+N \, ,
\end{equation}
where $M\equiv m^{\mu\nu}M_{\mu\nu}$ is the Ricci scalar associated to $m$ on $\mathcal{M}$ and $N\equiv n^{ab}N_{ab}$ is the Ricci scalar associated to $n$ on $\mathcal{N}$.
 \section{Schwarzschild-AdS Geodesic Flow}\label{geody}
In the following section, we are going to work with the standard form
\begin{equation}
\ds^{2}=g_{\mu\nu}\dx^{\mu}\dx^{\nu}=-F(r;M,\alpha)\dt^2 + \frac{1}{F(r;M,\alpha)}\diff r^2 + r^2 \diff\Omega^2_{2}
\end{equation}
of Schwarzschild-AdS metric on a spacetime manifold $\mathcal{M}$, where:
\begin{equation}
F(r;M,\alpha)=1 - \frac{2M}{r}+ \frac{r^2}{\alpha^{2}} \, .
\end{equation}
Considering an observable $h$ from $\mathcal{M}$ into real numbers, we define its spacetime average as
\begin{equation}\label{mean}
\Mean{h}\equiv\frac{1}{V_{\mathcal{M}}}\int_{\mathcal{M}}\sqrt{g}h \, ,
\end{equation}
where $g\equiv -\det{g_{\mu\nu}}=r^{2}\sin\theta$ and:
\begin{equation}
V_{\mathcal{M}}\equiv\int_{\mathcal{M}}\sqrt{g} \, .
\end{equation}
Following the usual procedure, we move to ingoing \textit{Eddington-Finklestein} coordinates
\begin{equation}
\ds^{2}=-F(r;M,\alpha)\diff v^2 + 2\diff v\diff r + r^2 \diff\Omega^2_{2} \, ,
\end{equation}
where:
\begin{equation}
\diff v\equiv\dt +\frac{\diff r}{F(r;M,\alpha)} \, .
\end{equation}
Now, we introduce a function $\mu\equiv\mu(M,\alpha)$ and rescale $r\equiv\mu\bar{r}$ and $v\equiv\mu\bar{v}$, obtaining:
\begin{equation}
\ds^{2}=-\mu^{2}F(\mu\bar{r};M,\alpha)\diff\bar{v}^2 + 2\mu^{2}\diff\bar{v}\diff\bar{r} + \mu^{2}\bar{r}^2 \diff\Omega^2_{2} \, .
\end{equation}
In order to compute geodesic equation, we introduce an affine parameter $\lambda$, use $\dot{k}$ notation for $\lambda$-derivative of an observable $k$, promote $M$ and $\alpha$ to $\lambda$-dependent parameters and define an auxiliary matrix:
\begin{equation}
f^{\alpha}_{\ \beta}\equiv g^{\alpha\sigma}\dot{g}_{\sigma\beta} \, .
\end{equation}
Using $f$, we can write geodesic equation as:
\begin{equation}
\dot{f}^{\alpha}_{\ \beta}=\frac{1}{4}\bigl[tr(f^{2})-\Mean{tr(f^{2})}\bigr]\delta^{\alpha}_{\ \beta}-\frac{1}{2}\bigl[tr(f)-\Mean{tr(f)}\bigr]f^{\alpha}_{\ \beta} \, .
\end{equation}
Now, we compute all terms separately. For the traces, we have:
\begin{equation}
tr(f)=f^{\alpha}_{\ \alpha}=\frac{8}{\mu}\dot{\mu}=8\frac{\diff}{\diff\lambda}\log{\mu[M(\lambda),\alpha(\lambda)]} \, ,
\end{equation}
\begin{equation}
tr(f^{2})=(f^{2})^{\alpha}_{\ \alpha}=\frac{16}{\mu^{2}}\dot{\mu}^{2}=\frac{1}{4}tr(f)^{2} \, .
\end{equation}
Being the above quantities constant over the whole spacetime, we clearly have:
\begin{equation}
\Mean{tr(f)}=tr(f);\qquad\Mean{tr(f^{2})}=tr(f^{2}) \, .
\end{equation}
Therefore, geodesic equation reduces to:
\begin{equation}
\dot{f}^{\alpha}_{\ \beta}=0 \, .
\end{equation}
Considering $\alpha=\beta=0$, we have:
\begin{equation}
\frac{\diff^{2}}{\diff\lambda^{2}}\log{\mu}=0\Longrightarrow\mu(\lambda)=e^{C_{0}(\lambda-\lambda_{0})+C_{1}} \, .
\end{equation}
By fixing $\lambda_{0}=0$, $M(\lambda)\equiv M_{\lambda}$, $\alpha(\lambda)\equiv\alpha_{\lambda}$ and $\mu[M_{0},\alpha_{0}]\equiv\mu_{0}$, we obtain:
\begin{equation}
\mu[M_{\lambda},\alpha_{\lambda}]\equiv\mu_{\lambda}=\mu_{0}e^{C_{0}\lambda} \, .
\end{equation}
Considering $\alpha=0$, $\beta=1$, we have:
\begin{equation}
\frac{\diff^{2}}{\diff\lambda^{2}}\biggl[2\frac{M}{\mu}-\bar{r}^{3}\frac{\mu^{2}}{\alpha^{2}}\biggr]=0\Longrightarrow 2\frac{M}{\mu}-\bar{r}^{3}\frac{\mu^{2}}{\alpha^{2}}=e^{C_{2}\lambda+C_{3}} \, .
\end{equation}
Once more, we impose initial conditions:
\begin{equation}
2\frac{M_{\lambda}}{\mu_{\lambda}}-\bar{r}^{3}\frac{\mu_{\lambda}^{2}}{\alpha_{\lambda}^{2}}=\biggl[2\frac{M_{0}}{\mu_{0}}-\bar{r}^{3}\frac{\mu_{0}^{2}}{\alpha_{0}^{2}}\biggr]e^{C_{2}\lambda} \, .
\end{equation}
Putting the two equations together, we have:
\begin{equation}
\frac{2}{\mu_{0}}\biggl[M_{\lambda}e^{-C_{0}\lambda}-M_{0}e^{C_{2}\lambda}\biggr]-\bar{r}^{3}\mu_{0}^{2}\biggl[\frac{1}{\alpha_{\lambda}^{2}}e^{2C_{0}\lambda}-\frac{1}{\alpha_{0}^{2}}e^{C_{2}\lambda}\biggr]=0 \, .
\end{equation}
Since this has to be fulfilled for any value of $\bar{r}$, we set the two terms equal to $0$ separately. Therefore, we get:
\begin{equation}
M_{\lambda}e^{-C_{0}\lambda}-M_{0}e^{C_{2}\lambda}=0 \, ,
\end{equation}
\begin{equation}
\frac{1}{\alpha_{\lambda}^{2}}e^{2C_{0}\lambda}-\frac{1}{\alpha_{0}^{2}}e^{C_{2}\lambda}=0 \, .
\end{equation}
Such equations can be trivially solved, providing us with:
\begin{equation}
M_{\lambda}=M_{0}e^{(C_{0}+C_{2})\lambda}\equiv M_{0}e^{A\lambda} \, ,
\end{equation}
\begin{equation}
\alpha_{\lambda}=\alpha_{0}e^{(C_{0}-C_{2}/2)\lambda}\equiv\alpha_{0}e^{B\lambda} \, .
\end{equation}
Since all other components of $f$ either gave the same equation as the one obtained from $f^{0}_{\ 0}$ or a trivial identity, we have two equations for two parameters. These lead to:
\begin{equation}
\begin{split}
r_{+}(\lambda)=&\frac{\biggl(9\alpha^{2}_{0}M_{0}+\sqrt{3}\sqrt{\alpha^{6}_{0}e^{2(B-A)\lambda}+27\alpha^{4}_{0}M^{2}_{0}}\biggr)^{1/3}}{3^{2/3}}e^{(A+2B)\lambda/3}+\\
&-\frac{\alpha^{2}_{0}e^{(4B-A)\lambda/3}}{3^{1/3}\biggl(9\alpha^{2}_{0}M_{0}+\sqrt{3}\sqrt{\alpha^{6}_{0}e^{2(B-A)\lambda}+27\alpha^{4}_{0}M^{2}_{0}}\biggr)^{1/3}} \, .
\end{split}
\end{equation}
We observe that, by imposing $A=B$, we obtain the simple behaviour:
\begin{equation}
r_{+}(\lambda)=r_{+}(0)e^{B\lambda} \, .
\end{equation}
Considering, once more, the general expression for $r_{+}(\lambda)$ and imposing $\alpha_{0}=1$, $M_{0}=C$, we simplify it as:
\begin{equation}
r_{+}(\lambda;A,B,C)=\frac{\biggl(9C+\sqrt{3}\sqrt{e^{2(B-A)\lambda}+27C^{2}}\biggr)^{1/3}}{3^{2/3}e^{-(A+2B)\lambda/3}}-\frac{e^{(4B-A)\lambda/3}}{3^{1/3}\biggl(9C+\sqrt{3}\sqrt{e^{2(B-A)\lambda}+27C^{2}}\biggr)^{1/3}} \, .
\end{equation}
\begin{figure}[H] 
  \begin{subfigure}[b]{0.5\linewidth}
    \centering
        \includegraphics[width=\linewidth]{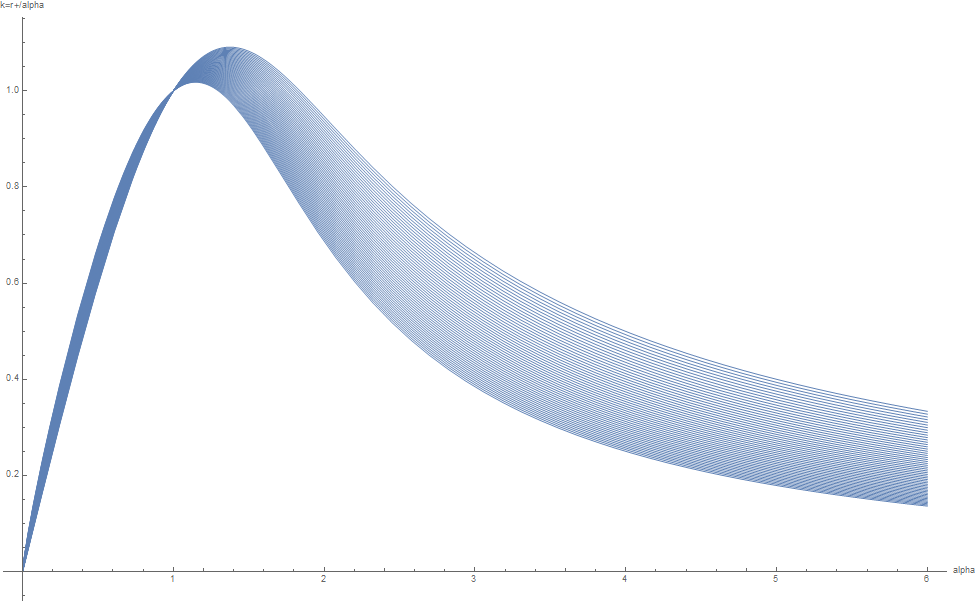}
    \caption{$C=1, B=1$ and $A\in[-2,5,-2]$}
    \label{fig7:a} 
    \vspace{4ex}
  \end{subfigure}
  \begin{subfigure}[b]{0.5\linewidth}
    \centering
    \includegraphics[width=\linewidth]{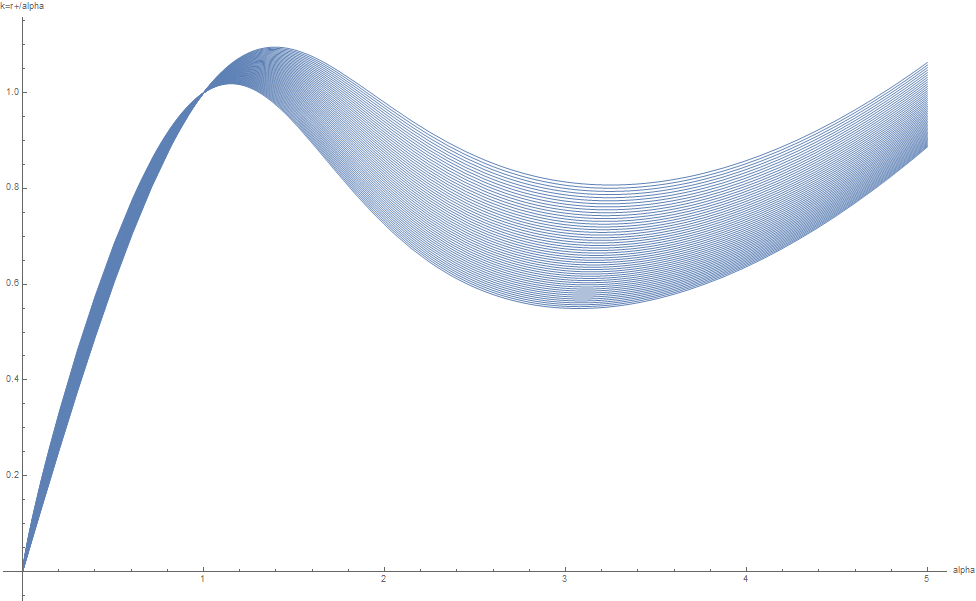}
    \caption{$C=1, B=0.995$ and $A\in[-2,5,-2]$}
    \label{fig7:b} 
    \vspace{4ex}
  \end{subfigure} 
  \begin{subfigure}[b]{0.5\linewidth}
    \centering
    \includegraphics[width=\linewidth]{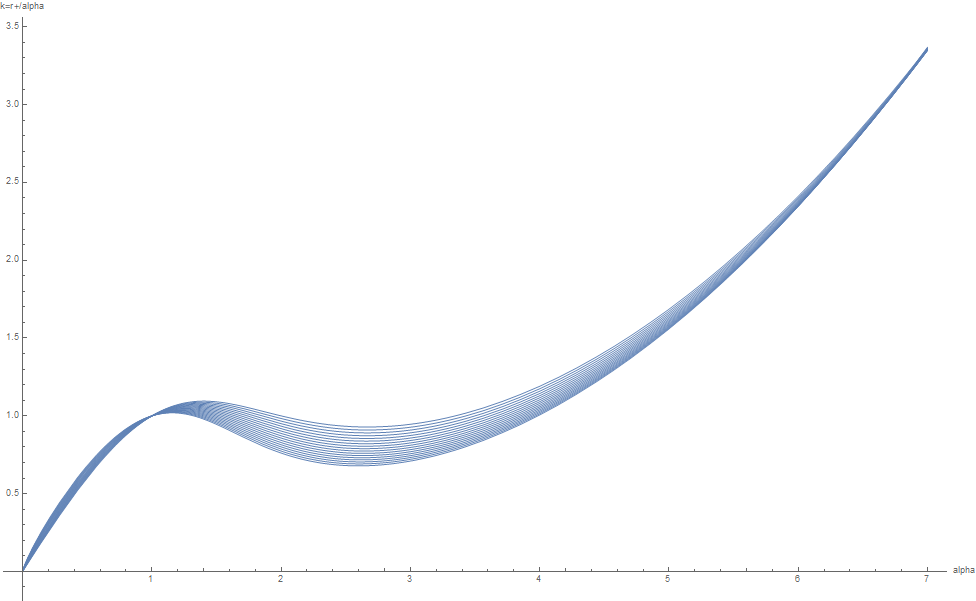}
    \caption{$C=1, B=0.99$ and $A\in[-2,5,-2]$}
    \label{fig7:c} 
  \end{subfigure}
  \begin{subfigure}[b]{0.5\linewidth}
    \centering
    \includegraphics[width=\linewidth]{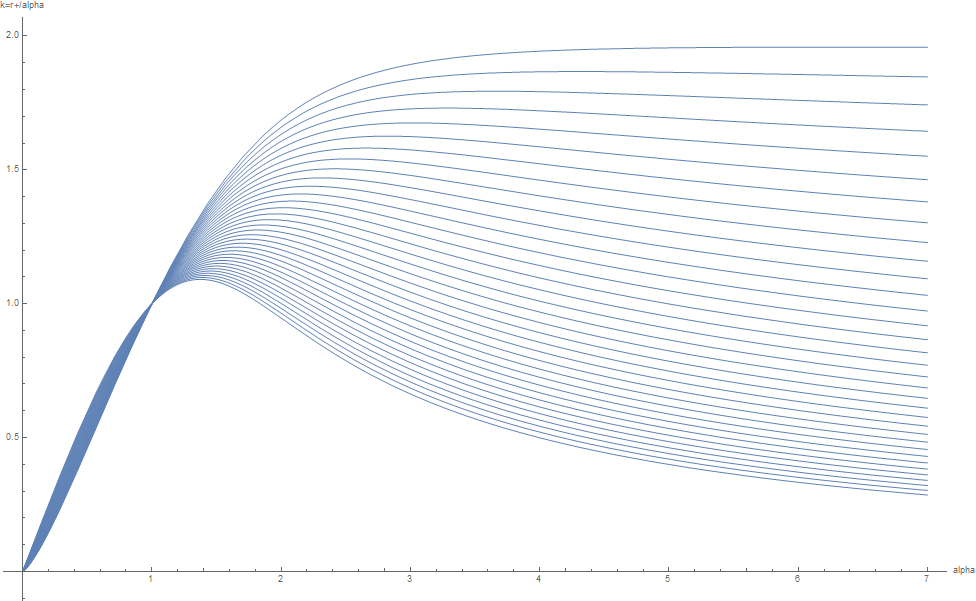}
    \caption{$C=1, B=1$ and $A\in[-2,-1]$}
    \label{fig7:d} 
  \end{subfigure} 
  \begin{subfigure}[b]{0.5\linewidth}
    \centering
        \includegraphics[width=\linewidth]{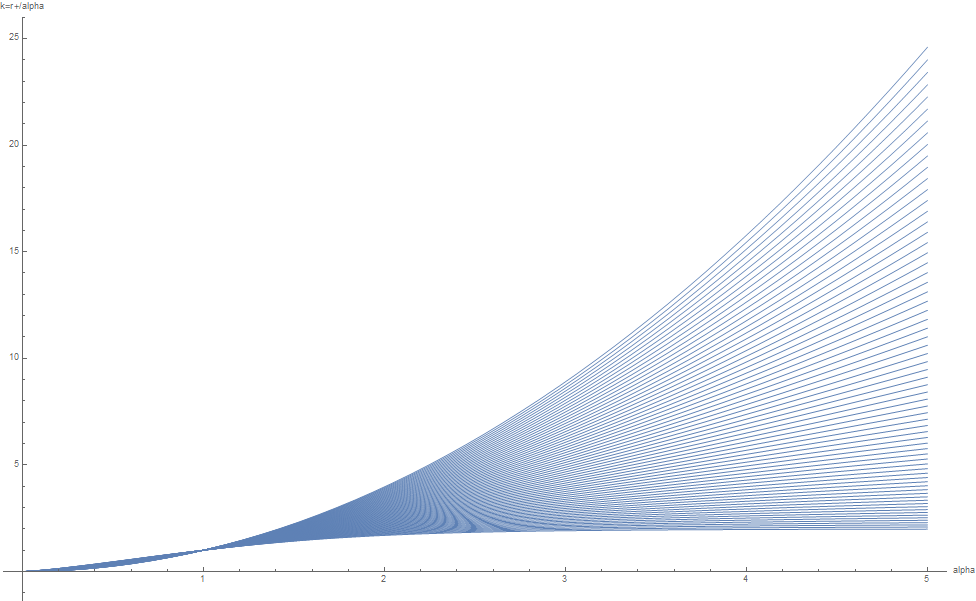}
    \caption{$C=1, B=1$ and $A\in[-1,1]$}
    \label{fig7:a} 
    \vspace{4ex}
  \end{subfigure}
  \begin{subfigure}[b]{0.5\linewidth}
    \centering
    \includegraphics[width=\linewidth]{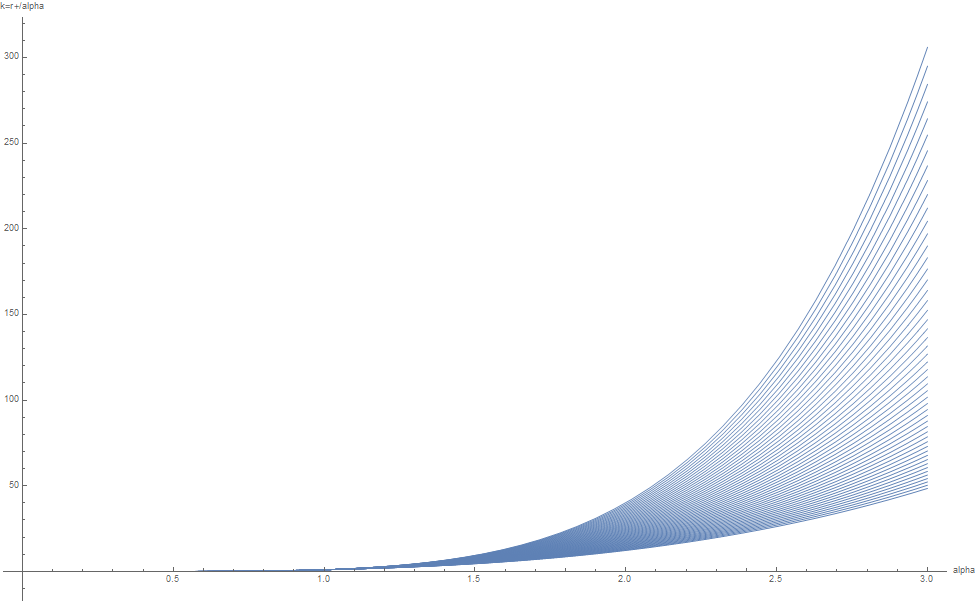}
    \caption{$C=1, B=1$ and $A\in[5,10]$}
    \label{fig7:b} 
    \vspace{4ex}
  \end{subfigure} 
  \begin{subfigure}[b]{0.5\linewidth}
    \centering
    \includegraphics[width=\linewidth]{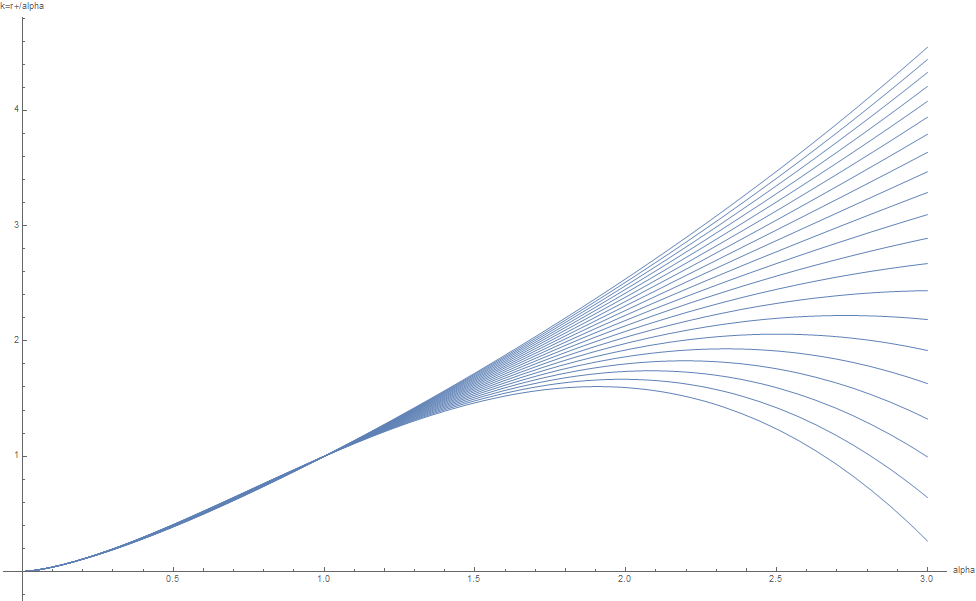}
    \caption{$C=1, A=-0.5$ and $B\in[0.9,1.1]$}
    \label{fig7:c} 
  \end{subfigure}
  \begin{subfigure}[b]{0.5\linewidth}
    \centering
    \includegraphics[width=\linewidth]{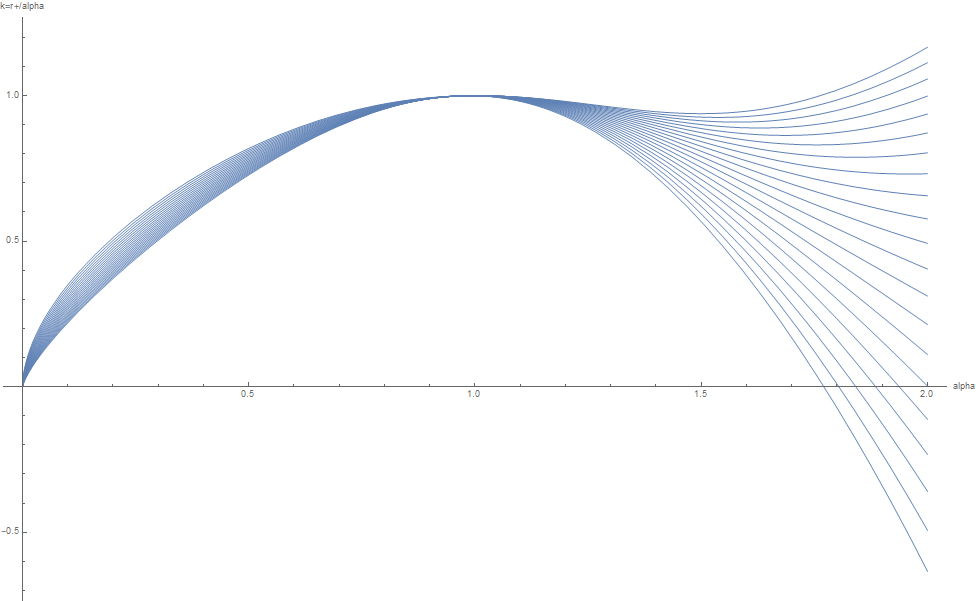}
    \caption{$C=1, A=-3$ and $B\in[0.9,1.1]$}
    \label{fig7:d} 
  \end{subfigure} 
\end{figure}
\newpage
\newpage
\end{appendices}

\bibliographystyle{plain}
\bibliography{bibliogra.bib}
\end{document}